\documentclass[11pt, reqno]{amsart}

\usepackage{color}
\usepackage{amsmath} 
\usepackage{amssymb}
\usepackage{amsthm}
\usepackage{amsxtra}
\usepackage{amscd}
\usepackage{bbm}
\usepackage{mathrsfs}
\usepackage{bm}
\usepackage[british]{babel}
\usepackage{braket}
\usepackage{latexsym}
\usepackage{amsmath}
\usepackage{dsfont} 
\usepackage{graphicx}
\usepackage{hyperref}

\newcommand{\Tr}{\mathrm{tr}}

\def\Z{\mathbb Z}
\def\N{\mathbb N}
\def\R{\mathbb R}

\def\i{{\mathrm i}}
\def\e{{\mathrm e}}

\def\epsilon{\varepsilon}
\def\eps{\varepsilon}

\newtheorem{thm}{\bfseries Theorem}
\newtheorem{lem}{\bfseries Lemma}
\newtheorem{cor}{\bfseries Corollary}
\newtheorem{conj}{\bfseries Conjecture}
\newtheorem{prop}{\bfseries Proposition}
\theoremstyle{definition}
\newtheorem{rem}{\bfseries Remark}
\newtheorem*{cor*}{\bfseries Corollary}
\usepackage{cite}

\begin{document}
\title[Excitation spectrum for weakly interacting bosons in a
trap]{The excitation spectrum for weakly interacting \\ bosons in a
  trap} \author{Philip Grech} \address{Department of Mathematics and
  Statistics, McGill University, 805 Sherbrooke Street West, Montreal,
  Quebec H3A 0B9, Canada; Centre de Recherches Math\'ematiques,
  Universit\'e de Montr\'eal, 2920 Chemin de la Tour, Montr\'eal,
  Qu\'ebec H3T 1J4, Canada.}  \email{pgrech@math.mcgill.ca}
\author{Robert Seiringer} \address{Department of Mathematics and
  Statistics, McGill University, 805 Sherbrooke Street West, Montreal,
  Quebec H3A 0B9, Canada.}  \email{rseiring@math.mcgill.ca}

\date{May 22, 2012}

\begin{abstract} We investigate the low-energy excitation spectrum of
  a Bose gas confined in a trap, with weak long-range repulsive
  interactions. In particular, we prove that the spectrum can be
  described in terms of the eigenvalues of an effective one-particle
  operator, as predicted by the Bogoliubov approximation.
\end{abstract}

\maketitle

\section{Introduction and main results}

\subsection{Introduction}

Bose-Einstein condensates of dilute atomic gases have been studied
extensively in recent years, both from an experimental and a
theoretical perspective \cite{DGPS,bloch}. Many fundamental aspects of
quantum mechanics were investigated  with the aid of these systems. One
of the manifestations of their quantum behavior is superfluidity,
leading to the appearance of quantized vortices in rotating systems
\cite{cooper,fetter}. This property is related to the structure of the
low-energy excitation spectrum, via the Landau criterion
\cite{landau}.  Excitation spectra of atomic Bose-Einstein condensates
have actually been measured \cite{Davidson}, and agreement was found
with theoretical predictions based on the Bogoliubov approximation
\cite{Bogo}.

From the point of view of mathematical physics, starting with the
basic underlying many-body Schr\"odinger equation, it remains a big
challenge to understand many features of cold quantum gases
\cite{LSSY,S}. While the validity of the Bogoliubov approximation for
evaluating the ground state energy has been studied in several cases
\cite{LSol,LSol2,Sol,ESY,GSlhy,yauyin}, no rigorous results on the
excitation spectrum of many-body systems with genuine interactions
among the particles are available, with the exception of certain
exactly solvable models in one dimension
\cite{gir,LL,L,calogero,sutherland}. In particular, it remains an open
problem to verify Landau's criterion for superfluidity in interacting
gases.

In this paper, we shall prove the accuracy of the Bogoliubov
approximation for the excitation spectrum of a trapped Bose gas, in
the mean-field or Hartree limit \cite{FL,FKS}, where the interaction
is weak and long-range. While the interactions among atoms in the
experiments on cold gases are more accurately modeled as strong and
short-range, effective long-range interactions can be achieved via
application of suitable electromagnetic fields \cite{Esslinger}. Our
work generalizes the recent results in \cite{Seiring}, where the
validity of Bogoliubov's approximation was verified for a homogeneous,
translation invariant model of interacting bosons. The inhomogeneity
caused by the trap complicates the analysis and leads to new features,
due to the non-commutativity of the various operators appearing in the
effective Bogoliubov Hamiltonian.

\subsection{Model and Main Results}\label{ss:mod}

We consider a system of $N\geq 2$ bosons in $\R^{d}$, for general
$d\geq 1$. The particles are confined by an external potential
$V_{\mathrm{ext}}(x)$, and interact via a weak two-body potential,
which we write as $(N-1)^{-1}v(x-y)$. The Hamiltonian of the system
reads, in suitable units,
\begin{align}\label{def:hn}
H_{N}=\sum_{i=1}^{N}\left(-\Delta_{i}+V_{\text{ext}}(x_{i})\right)+\frac{1}{N-1}\sum_{i<j}v(x_{i}-x_{j}) \, ,
\end{align}
with $\Delta$ denoting the standard Laplacian on $\R^d$. It acts on
the Hilbert space of permutation-symmetric square integrable functions
on $\R^{dN}$, as appropriate for bosons.  We assume that $v$ is a
bounded symmetric function, which is non-negative and of positive
type, i.e., has non-negative Fourier transform. The external potential $V_{\text{ext}}$ 
is assumed to be locally bounded and to satisfy
$V_{\mathrm{ext}}(x)\to \infty$ as $|x|\to \infty$.

Under these assumptions on $V_{\text{ext}}$ and $v$, the non-linear Hartree equation
\begin{align}\label{eq:hart}
(-\Delta +V_{\text{ext}})\varphi_{0} + (v \ast |\varphi_{0}|^{2})\varphi_{0}=\eps_{0} \varphi_{0}\, 
\end{align}
admits a unique strictly positive solution $\varphi_{0}$, normalized
as $\int \varphi_0^2 =1$, which is equal to the ground state of the
corresponding Hartree energy functional. In addition, there is a complete
set of normalized eigenfunctions $\{\varphi_{i}\}_{i\in \N}$ for the
Hartree operator
\begin{align}\label{def:HH}
H_{\text{H}}=-\Delta  +V_{\text{ext}}+v\ast \varphi_{0}^{2}\,.
\end{align}
The corresponding eigenvalues will be denoted by $\eps_{0}<
\eps_{1}\leq \eps_{2}\dots$ . We note that $\varphi_0$ is necessarily
the ground state of $H_{\rm H}$, since it is an eigenfunction that is
positive. Moreover, we emphasize that the inequality $\epsilon_1>
\epsilon_0$ is strict, since operators of the form (\ref{def:HH}) have
a unique ground state \cite{RS}. This will be essential for our
analysis.

Let $V$ denote the operator defined by the integral kernel
$$
V(x,y) = \varphi_0(x) v(x-y) \varphi_0(y) \, .
$$
As shown below, our assumptions on $v$ imply that this defines a
positive trace-class operator, whose trace is equal to $\Tr\, V = v(0)
= \|v\|_\infty$. Define also
\begin{equation}\label{def:D}
D: = H_{\text{H}} - \epsilon_0 = \sum_{i\geq 0}(\eps_{i}-\eps_{0})\ket{\varphi_{i}}\bra{\varphi_{i}}
\end{equation}
and let 
\begin{equation}\label{def:E}
E:= \left(D^{1/2}(D+2V)D^{1/2}\right)^{1/2}\, .
\end{equation}
Since $V$ is positive and bounded, $E$ is well-defined on the domain
of $D$.  We note that both $D$ and $E$ are, by construction, positive
operators, with $D\varphi_0 = E \varphi_0 = 0$. The Hartree minimizer
$\varphi_0$ is the only function in their kernel, all other
eigenvalues of $D$ and $E$ are strictly positive.


It turns out that $E-D-V$ is a trace class operator. (We will prove
this in Subsection~\ref{ss:II} below.) Let $0=e_0 < e_1 \leq e_2 \leq
\dots$ denote the eigenvalues of $E$. Our main result concerns the
spectrum of the Hamiltonian $H_N$, and reads as follows:

\begin{thm}\label{mainthm}
The ground state energy $E_{0}(N) = {\rm inf\, spec\,} H_{N}$ equals
\begin{align}\nonumber 
  E_{0}(N)& =N \int_{\R^d} \left( |\nabla \varphi_0(x)|^2 +
    V_{\rm{ext}}(x) \varphi_0(x)^2 \right) dx +\frac {N+1} 2
  \int_{\R^{2d}} \varphi_0(x)^2 v(x-y) \varphi_0(y)^2 dx\, dy \\ &
  \quad -\frac{1}{2}\mathrm{tr} (D+V-E) + O(N^{-1/2})\,
  .\label{gsenerg}
\end{align}
Moreover, the spectrum of $H_{N}-E_{0}(N)$ below an energy $\xi$ is equal to finite sums of the form
\begin{align}\label{spectrum}
\sum_{i\geq 1} e_{i}n_{i}+O(\xi^{3/2}N^{-1/2})\,,
\end{align}
where $n_{i}\in \N$ with $\sum_{i\geq 1} n_{i}\leq N$.
\end{thm}
The error term $O(N^{-1/2})$ in (\ref{gsenerg}) stands for an
expression which is bounded by a constant times $N^{-1/2}$ for large
$N$, where the constant only depends on the interaction potential $v$
and the gap $\eps_{1}-\eps_{0}$ in the spectrum of $H_\text{H}$;
likewise for the error term $O(\xi^{3/2}N^{-1/2})$ in
(\ref{spectrum}). The dependence on $v$ is relatively complicated but
could in principle be computed explicitly by following our proof; all our
bounds are quantitative.

Our result is a manifestation of the fact that the Bogoliubov
approximation becomes exact in the Hartree limit $N\to \infty$. In
particular, as long as $\xi \ll N$, each individual excitation energy
$\xi$ is of the form $\sum_{i\geq 1} e_{i}n_{i}(1+o(1))$. This result is
expected to be optimal in the following sense: if $\xi \ll N$ fails to
hold then there is a non-negligible number of particles outside the
condensate, violating a key assumption of Bogoliubov's approximation
\cite{Bogo,LSSY,Seiring}. Hence there is no reason why the Bogoliubov
approximation should predict the correct spectrum for excitation energies of order $N$ or larger.

Theorem~\ref{mainthm} states that the low-energy spectrum of $H_{N}-
E_0(N)$ is, up to small errors, equal to the one of the effective
operator
\begin{equation}\label{sumE}
\sum_{i=1}^N \hat E_i \quad , \quad \hat E = \sum_{j\geq 1} e_j |\varphi_j\rangle\langle \varphi_j| \,,
\end{equation}
where the subscript $i$ in $\hat E_i$ stands for the action of the
operator $\hat E$ on the $i$'th variable. Note that $\hat E$ is
unitarily equivalent to the operator $E$ defined in (\ref{def:E}).
The proof of Theorem~\ref{mainthm} actually consists of constructing
an explicit unitary operator that relates $H_N - E_0(N)$ and
(\ref{sumE}). In other words, we shall bound $H_{N}-E_0(N)$ from above
and below by a suitable unitary transform (cf. Eq.~(\ref{defofU})
below) of (\ref{sumE}), with error terms that are small in the
subspace of low energy.  As a byproduct of the proof we obtain the
following corollary.

\begin{cor} \label{cor2} Let $P_{H}^{j}$ be the projection onto the
  subspace spanned by the eigenfunctions corresponding to the $j$
  lowest eigenvalues of $H_{N}$ (counted with
  multiplicity). Similarly, let
  $P_{K}^{j}=\sum_{k=1}^{j}\ket{\psi_{k}}\bra{\psi_{k}}$ be the
  projection onto the subspace spanned by the eigenfunctions corresponding to the $j$ lowest
  eigenvalues of
\begin{align*}
  K:=\mathcal{U}^{\dagger}\left(\sum_{i=1}^N \hat E_i
  \right)\mathcal{U}+1=:\sum_{i=1}^{\infty}k_{i}\ket{\psi_{i}}\bra{\psi_{i}}
\end{align*}
($k_{1}\leq k_{2}\leq \dots$), where $\mathcal U$ is the unitary operator
defined in (\ref{defofU}). Then there is a constant $C$, depending
only on $v$ and $\eps_{1}- \eps_{0}$, such that if $k_{j+1}>k_{j}$ then
\begin{align*}
\| P_{K}^{j}-P_{H}^{j}\|_{2}^{2}\leq C(k_j/N)^{1/2}\frac{ \sum_{l=1}^{j}k_{l}}{k_{j+1}-k_{j}}\, ,
\end{align*}
with $\|\,\cdot\,\|_2$ denoting the Hilbert-Schmidt norm.
\end{cor}

The corollary implies, in particular, that the ground state wave function $\Psi_{0}$ of $H_{N}$ satisfies 
\begin{align}\label{lest}
\left\| \Psi_{0}- \mathcal{U}^{\dagger} \otimes_{i=1}^N \varphi_0\right\|^{2}\leq C N^{-1/2}
\end{align}
(for a suitable choice of the phase factor). The presence of the
unitary operator $\mathcal{U}$ in (\ref{lest}) is important, we do not
expect that $\Psi_{0}$ is close to $\otimes_{i=1}^N \varphi_0$ in an
$L^{2}$-sense for large $N$. (Compare with Remark~\ref{finrem} in
Section~\ref{corproof}.)

In addition, the corollary states that the eigenfunctions of $H_{N}$
near the bottom of the spectrum are approximately given by
$\mathcal{U}^\dagger$ applied to the eigenfunctions of (\ref{sumE}),
which are symmetrized products of the eigenfunctions $\varphi_i$ of
$H_\text{H}$ in (\ref{def:HH}). These functions can be obtained by
applying a number, $n$, of raising operators $a^\dagger(\varphi_i)$ to
the $N-n$ particle ground state, which is simply the product
$\prod_{i=1}^{N-n}\varphi_0(x_i)$. (Here we use the convenient Fock
space notation of creation operators, which will be recalled in the
next section.) In Subsection~\ref{ss:I}, we shall also calculate
$\mathcal{U}^\dagger a^\dagger(\varphi_i) \mathcal{U}$ (up to small
error terms), and hence arrive at a convenient alternative
characterization of the excited eigenstates of $H_N$.  (See
Remark~\ref{rem:pr} in Section~\ref{corproof}.)

\begin{rem}\label{nonrig}
  The emergence of the effective operator $E$ in (\ref{def:E}) can
  also be understood as follows. One considers the time-dependent
  Hartree equation $\i\partial_t \varphi = (-\Delta + V_{\rm ext} +
  |\varphi|^2*v) \varphi$ and looks for solutions of the form $\varphi
  = e^{-\i\eps_0 t}( \varphi_0 + u\, e^{-\i \omega t} + \overline{y}\,
  e^{\i\omega t})$ for some $\omega > 0$. Expanding to first order in
  $u$ and $y$ leads to the Bogoliubov-de-Gennes equations (see, e.g.,
  \cite{StringPit}, Eq.~(5.68))
\begin{align}\label{physeigprob}
  \left(\begin{array}{cc}D+V & V \\ -V &
      -(D+V)\end{array}\right)\left( \begin{array}{c} u \\
      y \end{array} \right) =\omega \left( \begin{array}{c} u \\
      y \end{array} \right) \, .
\end{align}
The positive values which can be assumed by $\omega$ are then
interpreted as excitation energies. This is in agreement with our
result: We will see below that the values for $\omega$ obtained this
way are precisely the eigenvalues of $E$. (Compare with
Remark~\ref{rem:bdg} in Section~\ref{sympl}.)
\end{rem}

\subsection{The translation-invariant case}

It is instructive to compare Theorem~\ref{mainthm} with the
translation invariant case studied in \cite{Seiring}, where the Bose
gas is confined to the flat unit torus $\mathbb{T}^{d}$. Up to an
additive constant, the Hartree operator equals the Laplacian in this
case, whose eigenfunctions are conveniently labeled by the quantized
momentum $p \in (2\pi \Z)^{d}$, and are given explicitly by the plane
waves $\varphi_{p}(x)=e^{\i p\cdot x}$. In this basis, the operators
$D$ and $V$ can be written as
\begin{align*}
D&=  \sum_{p \in (2\pi \Z)^{d} }   {p^{2}} \ket{\varphi_{p}}\bra{\varphi_{p}}  \\
V&=\sum_{p \in (2\pi \Z)^{d} }  \hat{v}(p) \ket{\varphi_{p}}\bra{\varphi_{p}}
\end{align*}
with $\hat{v}(p)= \int_{\mathbb{T}^{d}}v(x)e^{-\i p\cdot x} dx = \hat
v(-p)$. Since $D$ and $V$ commute in this case, we further have
\begin{align*}
E=  \sum_{p \in (2\pi \Z)^{d}} \sqrt{{{p^{4}}+{2p^{2}}\hat{v}(p)}} \ket{\varphi_{p}}\bra{\varphi_{p}} \,.
\end{align*}
Hence
\begin{align*}
  \mathrm{tr}\left(D+V - E\right) = \sum_{p \in (2\pi \Z)^{d}}
  \left(p^{2}+\hat{v}(p)-\sqrt{p^{4}+2p^{2}\hat{v}(p) } \right) \, ,
\end{align*}
and the eigenvalues of $E$ are given by 
\begin{align*}
e_{p} = \sqrt{p^{4}+2p^{2}\hat{v}(p) } \, ,
\end{align*}
yielding the well-known Bogoliubov spectrum of elementary excitations,
which is linear in $|p|$ for small momentum.

\subsection{Short-range interactions}

In Theorem~\ref{mainthm}, we assumed that $v(x)$ is a bounded
function. If we replace $v(x)$ by $g\delta(x)$, then $D+V-E$ will, in
general, fail to be trace class (in fact, it is not for the above
model of bosons on $\mathbb{T}^{d}$ for $d\geq 2$). However,
Formula~(\ref{spectrum}) for the excitation spectrum still makes
sense. Since all our bounds are quantitative, our proof thus shows
that if $v$ is allowed to depend on $N$ in such a way that it
converges to a $\delta$-function, and $v(0)$ increases with $N$ {\em
  slow enough}, then the excitation spectrum is still of the form
$\sum_{i} e_i n_i$, where $e_i$ are the non-zero eigenvalues of $E$ in
(\ref{def:E}), and $V$ is now the multiplication operator $g
\varphi_0(x)^2$. If $v(0)$ increases too fast with $N$, though, our
error bounds cease to be good enough to allow this
conclusion. 

Consider now the case $d=3$. If we write the interaction potential as $(N-1)^{-1} \lambda_N^3
v_0(\lambda_N x)$ for some fixed, $N$-independent $v_0$, with
$\lambda_N\to \infty$ as $N\to \infty$, we expect that the Bogoliubov
approximation yields the correct excitation spectrum as long as
$\lambda_N \ll N$.  If $\lambda_N\sim N$, the scattering length of the
interaction potential is of the same order as the range of the
interactions. This corresponds to the Gross-Pitaevskii scaling
\cite{LSSY} of a dilute gas. In this latter case, the scattering length
becomes the physically relevant parameter quantifying the interacting
strength, instead of $\int_{\R^{3}} v(x) dx$. Hence we expect the
following to be true.

\begin{conj}\label{gpconj}
Consider the Hamiltonian
\begin{align*}
  H_{N}^{\mathrm{GP}}:=\sum_{i=1}^{N}\left(-\Delta_{i}+V_{\rm{ext}}(x_{i})\right)+N^{2}\sum_{i<j}v(N(x_{i}-x_{j}))
  \, ,
\end{align*}
on $L^{2}(\R^{3 })^{\otimes_{\text{s}}N}$, with $v$ non-negative,
bounded and integrable at infinity, and denote its ground state energy
by $E_{0}(N)$. The spectrum of $H_{N}^{\mathrm{GP}}-E_{0}(N)$ below an
energy $\xi\ll N $ is equal to finite sums of the form
\begin{align*}
  \sum_{i\geq 1} e_{i}n_{i} \left(1 + o(1) \right)
\end{align*}
for large $N$. Here, $e_{i}$ is defined as in the Hartree case with
the replacements
\begin{align*}
H_{\rm H}&\leadsto H_{\mathrm{GP}}:=-\Delta+V_{\mathrm{ext}}+8\pi a_{0}\varphi_{0}^{2} \\
V&\leadsto 8\pi a_{0} \varphi_0^2 \, ,
\end{align*}
where $\varphi_0$ is now the minimizer of the Gross-Pitaevskii energy
functional, and $a_{0}$ is the zero energy scattering length of the
interaction potential $v(x)$.
\end{conj} 

We expect the proof of Conjecture~\ref{gpconj} to be more complicated
than that of Theorem~\ref{mainthm}. In particular, the Bogoliubov
approximation would have to be modified in such a way to account for
the detailed structure of the wave function when particles are close,
which gives rise to the scattering length $a_{0}$ (instead of merely
its first-order Born approximation $(8\pi)^{-1}\int_{\R^{3}} v(x)
dx$).

\subsection{Outline}

The remainder of the paper is organized as follows. In
Section~\ref{prelim} we establish bounds on the number of particles
outside the condensate, the $N$-body Hartree operator $\sum_{i=1}^N
D_i$, and their product for a low-energy state. Section~\ref{Bogoham}
shows how $H_{N}$ can be bounded from above and below by what we call
the Bogoliubov Hamiltonian, which is formally close to Bogoliubov's
approximate quadratic Hamiltonian on Fock space, yet is particle
number conserving. The diagonalization of the quadratic Hamiltonian
can be achieved by a Bogoliubov transformation, which is carried out
in Section~\ref{sympl}. To diagonalize the actual Bogoliubov
Hamiltonian we use a modification thereof, which involves the
estimation of various error terms (Section~\ref{sec_est}). Finally, we
shall complete the proof of Theorem~\ref{mainthm}
(Section~\ref{finish}) and Corollary~\ref{cor2}
(Section~\ref{corproof}).

Throughout this work a multiplicative constant $C$ in an estimate is
understood to be generic: it can have different values on each
appearance. By $\|\,\cdot\,\|$ we denote the operator or vector norm,
depending on context; $\|\cdot \|_{1}$ and $ \|\cdot \|_{2}$ denote
the trace class and Hilbert-Schmidt norms of operators, respectively.

\section{Bounds on the Condensate Depletion}\label{prelim}

It is convenient to regard the $N$-particle Hilbert space
$\mathcal{F}^{(N)}:=L^{2}(\R^{d})^{\otimes_{\text{s}}N}$, the
symmetric tensor product of $N$ one-particle Hilbert spaces
$L^{2}(\R^{d})$, as a subspace of the bosonic Fock space $\mathcal{F}=
\oplus_{N=0}^{\infty}\mathcal{F}^{(N)}$. The Hamiltonian $H_{N}$ can
then be written in second quantized form as
\begin{align}\label{hnsq}
  H_{N}=\sum_{i,j}h_{ij}a_{i}^{\dagger}a_{j}+\frac{1}{2(N-1)}\sum_{i,j,k,l}v_{ijkl}a_{j}^{\dagger}a_{i}^{\dagger}a_{k}a_{l}
\end{align}
where
$$
h_{ij}:=\langle \varphi_{i} | -\Delta+V_{\text{ext}} |
\varphi_{j}\rangle \quad \text{and} \quad v_{ijkl}:=\langle
\varphi_{i},\varphi_{j}|v| \varphi_{k},\varphi_{l}\rangle \,.
$$
Recall that the set $\{\varphi_i\}_{i\in \N}$ denotes the orthonormal
basis of eigenfunctions of $H_{\text{H}}$ in (\ref{def:HH}), which we
can all assume to be chosen {\em real} without loss of generality. The
operators $a_{i}^\dagger$ and $a_{i}$ in (\ref{hnsq}) are the usual
creation and annihilation operators corresponding to these functions,
i.e., $a_{i}:=a(\varphi_{i})$.

To be precise, $H_N$ in (\ref{def:hn}) agrees with the right side of
(\ref{hnsq}) on the subspace $\mathcal{F}^{(N)}$. We shall always work
on this subspace, and use Fock space notation only for convenience. In
particular, unless stated otherwise, all subsequent identities and
inequalities involving operators on Fock space are understood as
holding on $\mathcal{F}^{(N)}$ only.

We introduce the rank-one projection
$P=\ket{\varphi_{0}}\bra{\varphi_{0}}$ and the complementary
projection $Q=1 -P$. The operator that counts the number of particles
outside the Hartree ground state is the second quantization
$\mathrm{d}\Gamma(Q)$ of $Q$ and will be denoted by $N^{>}$, i.e.
\begin{align*}
  N^{>}=\sum_{i=1}^{N}Q_{i}=\sum_{i}{}^{'}a^{\dagger}_{i}a_{i}\, .
\end{align*}
Here and in the following, $\sum^{'}$ denotes a sum over all nonzero
indices.  Another important quantity is the following $N$-body Hartree
operator,
\begin{align}\label{def:TH}
  T_\text{H}:=\sum_{i=1}^{N}\left(-\Delta_{i}+V_{\text{ext}}(x_{i})+(v\ast
    \varphi_{0}^{2})(x_{i})-\eps_{0}\right) = \mathrm{d}\Gamma(D)\, ,
\end{align}
with $D$ defined in (\ref{def:D}).

The following lemma gives simple bounds on the ground state energy of
$H_N$, as well as on the expectation values of $N^{>}$ and
$T_\text{H}$ in low-energy states.

\begin{lem}\label{simple bounds} The ground state energy $E_0(N)$ of
  $H_{N}$ satisfies the bounds
  \begin{align*}
    0\geq E_{0}(N)-Nh_{00}-\frac{N}{2}v_{0000}\geq{\frac 1
      2}v_{0000}-\frac{N}{2(N-1)}v(0)\, .
  \end{align*}
  Moreover, for any $N$-particle state $\Psi$ with $\bra{\Psi} H_{N}
  \ket{\Psi}\leq Nh_{00} +\frac{N}{2}{v}_{0000}+\mu$, we have
  \begin{align}\label{expvalineq}
    (\eps_{1}-\eps_{0})\bra{\Psi} N^{>} \ket{\Psi} \leq
    \bra{\Psi}T_\text{H}\ket{\Psi}\leq \mu+
    \frac{N}{2(N-1)}v(0)-\frac{v_{0000}}{2}\, .
  \end{align}
\end{lem}

Recall that $\epsilon_0$ and $\epsilon_1$ denote the lowest two
eigenvalues of the Hartree operator $H_{\rm H}$ in (\ref{def:HH}).  We
emphasize that $\eps_{1}-\eps_{0}> 0$.

\begin{proof}
  For the upper bound we use the trial function $\ket{ N,0,\dots}$
  denoting a state where all particles occupy the ground state of the
  Hartree operator $H_{\text{H}}$. This yields
  \begin{align*}
    E_{0}(N)&\leq \sum_{i,j}h_{ij} \bra{N,0,\dots} a_{i}^{\dagger}a_{j}\ket{N,0,\dots} \\ &\ \ +\frac{1}{2(N-1)}\sum_{i,j,k,l}v_{ijkl} \bra{N,0,\dots} a_{j}^{\dagger}a_{i}^{\dagger}a_{k}a_{l}\ket{N,0,\dots} \\
    &= Nh_{00}+\frac{N}{2}v_{0000}\, .
  \end{align*}
  For the lower bound we exploit the positive definiteness of the
  interaction potential $v$ in the following way. With
  $\psi(x)=\varphi_{0}^{2}(x)-\frac{1}{{N-1}}\sum_{i=1}^{N}\delta(x-x_{i})$,
  we have
  \begin{align} \nonumber 0 & \leq
    \int_{\R^{2d}}\psi(x)v(x-y)\psi(y)dxdy \\ & =
    v_{0000}-\frac{2}{N-1}\sum_{i=1}^{N}(v \ast
    \varphi_{0}^{2})(x_{i})+\frac{1}{(N-1)^{2}}\sum_{i,j}v(x_{i}-x_{j})\,
    .  \label{posest}
  \end{align}
  Put differently, this inequality reads
  \begin{align}\label{interacts}
    \frac{1}{N-1}\sum_{i<j}v(x_{i}-x_{j}) \geq
    -\frac{N-1}{2}v_{0000}+\sum_{i=1}^{N}(v\ast
    \varphi_{0}^{2})(x_{i})-\frac{N}{2(N-1)}v(0)\,.
  \end{align}
  Since $\eps_0 = h_{00} + v_{0000}$, we hence have
  \begin{align*}
    H_{N}&\geq \sum_{i=1}^{N}\left(-\Delta_{i}+V_{\text{ext}}(x_{i})+(v\ast \varphi_{0}^{2})(x_{i})\right) -\frac{N-1}{2}v_{0000}-\frac{N}{2(N-1)}v(0) \nonumber \\
    &= T_{\text{H}}+Nh_{00}+\frac{N+1}{2}v_{0000}-\frac{N}{2(N-1)}v(0)
    \,.
  \end{align*}
  The asserted bounds now follows, since $T_{\text{H}} \geq
  (\epsilon_1-\epsilon_0) N^> \geq 0$.
\end{proof}

\begin{rem} \label{opin1} The proof actually shows the operator
  inequality
  \begin{align*}
    T_\text{H}\leq H_{N}-Nh_{00}-\frac{N+1}{2}v_{0000}
    +\frac{N}{2(N-1)}v(0)
  \end{align*}
  from which (\ref{expvalineq}) readily follows.
\end{rem}

In our analysis we shall also need bounds on the expectation value of
the product $N^{>}T_\text{H}$ for a low-energy state. Such a bound is
the content of Lemma~\ref{quadrbdlem}.

\begin{lem} \label{quadrbdlem} Let $\Psi$ be an $N$-particle wave
  function in the spectral subspace of $H_{N}$ corresponding to an
  energy $E\leq E_{0}(N)+\mu$. Then
  \begin{align*}
    (\eps_{1}-\eps_{0})\bra{\Psi} N^{>}T_\text{H}\ket{\Psi} \leq
    (\mu-v_{0000}&+3v(0))\left( \mu+
      \frac{N}{2(N-1)}v(0)-\frac{v_{0000}}{2}\right) +\frac{1}{4}
    {\left(2v(0)+\mu\right)^{2}}\, .
  \end{align*}
\end{lem}
\begin{rem}
  A slight modification of the proof yields the operator inequality
  \begin{align*}
    (\eps_{1}-\eps_{0})N^{>}T_\text{H}\leq
    (3v(0)-v_{0000})T_\text{H}+2v(0)^{2}+2(H_{N}-E_{0}(N))^{2}\, .
  \end{align*}
\end{rem}

\begin{proof} We write
  \begin{align*}
    \bra{\Psi} N^{>}T_\text{H}\ket{\Psi} = \bra{\Psi} N^{>}S
    \ket{\Psi} + \left\langle \Psi \left| N^{>}\left(
          H_{N}-E_{0}(N)-\frac{\mu}{2}\right) \right| \Psi
    \right\rangle
  \end{align*}
  where $S= \sum_{i=1}^{N} (v\ast \varphi_{0}^{2})(x_{i})- N\eps_{0}
  -\frac{1}{N-1}\sum_{i<j}v(x_{i}-x_{j}) +\frac{\mu}{2}+
  E_{0}(N)$. The second term can be bounded by Schwarz's inequality as
  \begin{align*}
    \left| \left\langle \Psi \left| N^{>}\left(
            H_{N}-E_{0}(N)-\frac{\mu}{2}\right) \right| \Psi
      \right\rangle \right| \leq \frac{\mu}{2} \bra{\Psi} (N^{>})^{2}
    \ket{\Psi}^{1/2}\, .
  \end{align*}
  For the first we use the permutation symmetry of $\Psi$ and get
  \begin{align*}
    \bra{\Psi} N^{>}S \ket{\Psi}=N\bra{\Psi}Q_{1}S\ket{\Psi} \, .
  \end{align*}
  We split $S$ into two parts, $S=S_{a}+S_{b}$, where
  \begin{align*}
    S_{a}&:= \sum_{i=2}^{N}(v\ast \varphi_{0}^{2})(x_{i})-N\eps_{0}-\frac{1}{N-1}\sum_{2\leq i <j}v(x_{i}-x_{j}) +\frac{\mu}{2} + E_{0}(N)\, , \\
    S_{b}&:= (v\ast
    \varphi_{0}^{2})(x_{1})-\frac{1}{N-1}\sum_{i=2}^{N}v(x_{1}-x_{i})\,
    .
  \end{align*}
  Using the positive definiteness of $v$ as in (\ref{posest}), this
  time for
  $\psi(x)=\varphi_{0}^{2}(x)-\frac{1}{{N-1}}\sum_{i=2}^{N}\delta(x-x_{i})$,
  we obtain
  \begin{align*}
    \frac{1}{N-1}\sum_{2\leq i<j}v(x_{i}-x_{j})\geq
    -\frac{(N-1)}{2}v_{0000}+\sum_{i=2}^{N}(v\ast
    \varphi_{0}^{2})(x_{i})-\frac{v(0)}{2} \,.
  \end{align*}
  In combination with the upper bound on $E_0(N)$ in Lemma~\ref{simple
    bounds} this implies that
  \begin{align*}
    S_{a}\leq \frac 1 2 \left( v(0)-v_{0000}+\mu \right)\, .
  \end{align*}
  In particular, since $S_a$ commutes with $Q_1$, we have
  \begin{align*}
    N\bra{\Psi} Q_{1}S_{a}\ket{\Psi} \leq \frac 1 2 \left(
      v(0)-v_{0000}+\mu \right)\bra{\Psi}N^{>}\ket{\Psi}\, .
  \end{align*}
  To bound the contribution of $S_{b}$, we compute
  \begin{align*}
    \bra{\Psi} Q_{1}S_{b} \ket{\Psi}=&\bra{\Psi}Q_{1}\left[(v\ast \varphi_{0}^{2})(x_{1})-v(x_{1}-x_{2})\right]\ket{\Psi} \\
    =&  \bra{\Psi} Q_{1}Q_{2}\left[(v\ast \varphi_{0}^{2})(x_{1})-v(x_{1}-x_{2})\right] \ket{\Psi} \\
    & +\bra{\Psi} Q_{1}P_{2}\left[(v\ast \varphi_{0}^{2})(x_{1})-v(x_{1}-x_{2})\right]P_{2}\ket{\Psi} \\
    & + \bra{\Psi} Q_{1}P_{2}\left[(v\ast
      \varphi_{0}^{2})(x_{1})-v(x_{1}-x_{2})\right]Q_{2}\ket{\Psi}\,.
  \end{align*}
  The second term on the right side of the last equation vanishes. For
  the first and the third, we use Schwarz's inequality and
  $|(v\ast\varphi_{0}^{2}(x_{1})-v(x_{1}-x_{2})|\leq v(0)$ to conclude
  \begin{align*}
    |\bra{\Psi} Q_{1}S_{b}\ket{\Psi}| \leq v(0) \bra{\Psi}
    Q_{1}Q_{2}\ket{\Psi}^{{1/2}} +v(0) \bra{\Psi} Q_{1}\ket{\Psi}\,.
  \end{align*}
  Since
  \begin{align*}
    N^{2} \bra{\Psi} Q_{1}Q_{2}\ket{\Psi} \leq \bra{\Psi} (N^{>})^{2}
    \ket{\Psi}
  \end{align*}
  we have thus shown that
  \begin{align*}
    \bra{\Psi} N^{>}T_\text{H}\ket{\Psi} \leq & \,\frac{1}{2} \left(\mu - v_{0000}+3v(0) \right) \bra{\Psi} N^{>}\ket{\Psi} \\
    &+\frac 1 2 \left(2v(0)+{\mu} \right) \bra{\Psi}
    (N^{>})^{2}\ket{\Psi}^{{1/2}}\, .
  \end{align*}
  Using $T_\text{H} \geq (\eps_{1}-\eps_{0}) N^{>}$ this implies
  \begin{align*}
    \bra{\Psi}N^{>}T_\text{H}\ket{\Psi} \leq
    (\mu-v_{0000}+3v(0))\bra{\Psi} N^{>}\ket{\Psi}+\frac{1}{4}
    \frac{(2v(0)+\mu)^{2}}{\eps_{1}-\eps_{0}}\, .
  \end{align*}
  The result then follows from Lemma~\ref{simple bounds}.
\end{proof}

\section{The Bogoliubov Hamiltonian}\label{Bogoham}

The well-known Bogoliubov approximation \cite{Bogo} consists of
replacing the operators $a_0$ and $a_0^\dagger$ in (\ref{hnsq}) by
$\sqrt{N}$, and dropping all terms higher than quadratic in the $a_i$
and $a^\dagger_i$ for $i\geq 1$. The resulting Bogoliubov Hamiltonian
does not preserve particle number and is thus not suitable as an
approximation to the full Hamiltonian $H_{N}$, as far as operator
inequalities are concerned. To circumvent this problem, we work with
the following modification of the Bogoliubov Hamiltonian. For $i\geq
1$, we introduce the operators
\begin{align*}
  b_{i}:= \frac{a_{i}a_{0}^{\dagger}}{\sqrt{N-1}}\, ,
\end{align*}
and we define the Bogoliubov Hamiltonian as
\begin{align}\label{def:hbog}
  H_{\text{Bog}}:= \sum_{i}{}^{'}
  \left(\eps_{i}-\eps_{0}\right)b_{i}^{\dagger}b_{i}+{\frac 1 2}
  \sum_{i,j}
  {}^{'}V_{ij}\left(2b_{i}^{\dagger}b_{j}+b_{i}b_{j}+b_{j}^{\dagger}b_{i}^{\dagger}\right)\,,
\end{align}
where $V_{ij}= v_{00ij} = \langle\varphi_i | V |\varphi_j\rangle$.
Note that this operator preserves the number of particles, hence we
can study its restriction to $\mathcal{F}^{(N)}$, the sector of $N$
particles. The price to pay, as compared with the usual Bogoliubov
Hamiltonian, is that the $b_{i}, b_{i}^{\dagger}$ do not satisfy
canonical commutation relations, making it harder to determine the
spectrum of $H_{\mathrm{Bog}}$.

In the following, we shall investigate the relation between $H_N$ and
$H_{\rm Bog}$. In particular, we shall derive upper and lower bounds
on $H_N$ in terms of $H_{\rm Bog}$, with error terms that are small in
the low-energy sector.

\subsection{Lower Bound}

Using the positivity of the interaction potential $v$, a Schwarz
inequality on $\mathcal{F}^{(2)}$ yields
\begin{align*}
  &(P\otimes Q+Q\otimes P)vQ\otimes Q +Q\otimes Q v (P\otimes Q+Q\otimes P) \\
  &\ \ \geq -\eps (P\otimes Q+Q\otimes P)v(P\otimes Q+Q\otimes
  P)-\eps^{-1}Q\otimes Qv Q\otimes Q \,.
\end{align*}
Consequently,
\begin{align}\nonumber
  v&\geq P\otimes P v P\otimes P +P \otimes P v Q\otimes Q + Q\otimes
  Q v P\otimes P \\ \nonumber
  &\quad +(1-\eps)(P\otimes Q+Q \otimes P)v(P\otimes Q +Q\otimes P)+(1-\eps^{-1})Q\otimes Qv Q\otimes Q \\
  &\quad +P\otimes P v P\otimes Q+ P\otimes P v Q\otimes P + P\otimes
  Q v P\otimes P +Q\otimes P v P \otimes P \label{bv}
\end{align}
for any $\eps>0$. The last term in the second line can be bounded from
below by $(1-\eps^{-1})v(0)Q\otimes Q$ as long as $\eps\leq 1$ which
we shall assume henceforth. We remark that in the case of translation
invariance the terms in the last line vanish due to momentum
conservation, but this is not the case here.

In second quantized language, the bound (\ref{bv}) implies that
$H_{N}$ is bounded from below by the operator
\begin{align}\label{lowerbound}
  \sum_{i,j}{}^{'}&h_{ij}a^{\dagger}_{i}a_{j}+\sqrt{N-1}\sum_{i}{}^{'}h_{i0}\left(b_{i}^{\dagger}+b_{i}\right) +h_{00}(N-N^{>})\nonumber \\
  &+v_{0000}\frac{(N-N^{>})(N-N^{>}-1)}{2(N-1)}+\frac1 2\sum_{i,j}{}^{'}V_{ij}\left(b_{i}b_{j}+b_{j}^{\dagger}b_{i}^{\dagger}\right) \nonumber\\
  &+\frac{1-\eps}{N-1} \sum_{i,j}{}^{'}\left(v_{0i0j}+ V_{ij}\right)a_{i}^{\dagger}a_{0}^{\dagger}a_{j}a_{0} \nonumber \\
  &+(1-\eps^{-1})\frac{N^{>}(N^{>}-1)v(0)}{2(N-1)}+\sum_{i}{}^{'}\frac{v_{i000}}{\sqrt{N-1}}\left(b_{i}(N-N^{>})+(N-N^{>})b_{i}^{\dagger}\right)
  \,,
\end{align}
restricted to the $N$-particle sector. We note that
\begin{align*}
  \sum_{i,j}{}^{'}(h_{ij}-\eps_{0}\delta_{ij})b_{i}^{\dagger}b_{j}&=\sum_{i,j}{}^{'}(h_{ij}-\eps_{0}\delta_{ij})a_{i}^{\dagger}a_{j}\frac{N-N^{>}+1}{N-1} \\
  &=
  \sum_{i,j}{}^{'}(h_{ij}-\eps_{0}\delta_{ij})a_{i}^{\dagger}a_{j}+\frac{2-N^>}{N-1}\sum_{i,j}{}^{'}(h_{ij}-\eps_{0}\delta_{ij})a_{i}^{\dagger}a_{j}
  \,.
\end{align*}
Since $D - v(0) \leq -\Delta + V_\text{ext} - \eps_0 \leq D$, we can
bound the last term as
$$
\frac{2-N^>}{N-1}\sum_{i,j}{}^{'}(h_{ij}-\eps_{0}\delta_{ij})a_{i}^{\dagger}a_{j}
\leq \frac 1{N-1} T_\text{H} + v(0) \frac{\left(N^>\right)^2}{N} \,.
$$
This bound can be easily verified by investigating separately the
sectors of different values of $N^>$. (In particular, note that
$T_\text{H}=0$ on the subspace where $N^{>}=0$, for instance.)

We also have
\begin{align*}
  &\frac{1-\eps}{N-1} \sum_{i,j}{}^{'}\left(v_{0i0j}+ V_{ij}\right)a_{i}^{\dagger}a_{0}^{\dagger}a_{j}a_{0} \\ &= \sum_{i,j}{}^{'}\left(v_{0i0j}+ V_{ij}\right)b_{i}^{\dagger}b_{j}-\frac{1+\eps(N-N^{>})}{N-1}   \sum_{i,j}{}^{'}\left(v_{0i0j}+ V_{ij}\right)a_{i}^{\dagger}a_{j} \\
  &\geq \sum_{i,j}{}^{'}\left(v_{0i0j}+
    V_{ij}\right)b_{i}^{\dagger}b_{j}-2v(0)\frac{N^{>}+\eps
    N^{>}(N-N^{>})}{N-1} \, ,
\end{align*}
where we have used that $V$ as well as multiplication with $v\ast
\varphi_{0}^{2}$ are bounded operators with norm bounded by
$v(0)$. Using $\eps_{0}=h_{00}+v_{0000}$ one verifies that
\begin{align*}
  \eps_{0}N^{>}+h_{00}(N-N^{>})+v_{0000}\frac{(N-N^{>})(N-N^{>}-1)}{2(N-1)}+(1-\eps^{-1})\frac{N^{>}(N^{>}-1)v(0)}{2(N-1)} \\
  = Nh_{00}+\frac N 2
  v_{0000}+((1-\eps^{-1})v(0)+v_{0000})\frac{N^{>}(N^{>}-1)}{2(N-1)}\,.
\end{align*}
The Hartree equation (\ref{eq:hart}) implies $h_{i0}+v_{i000}=0$ for
$i\neq 0$, hence we have
\begin{align*}
  \sqrt{N-1}\sum_{i}{}^{'}h_{i0}\left(b_{i}^{\dagger}+b_{i}\right)+\sum_{i}{}^{'}\frac{v_{i000}}{\sqrt{N-1}}\left(b_{i}(N-N^{>})+(N-N^{>})b_{i}^{\dagger}\right)
  \\
  =\sum_{i}{}^{'}\frac{v_{i000}}{\sqrt{N-1}}\left(b_{i}(1-N^{>})+(1-N^{>})b_{i}^{\dagger}\right)
  \, .
\end{align*}
This last expression can be bounded by Schwarz's inequality: for any
$\zeta>0$ one has
\begin{align*}
  -\frac{(1-N^{>})^{2}}{\zeta\sqrt{N-1}}-\zeta
  v(0)^{2}\frac{NN^{>}}{(N-1)^{3/2}} & \leq
  \sum_{i}{}^{'}\frac{v_{i000}}{\sqrt{N-1}}\left(b_{i}(1-N^{>})+(1-N^{>})b_{i}^{\dagger}\right) \\
  & \leq \frac{(1-N^{>})^{2}}{\zeta\sqrt{N-1}}+\zeta
  v(0)^{2}\frac{NN^{>}}{(N-1)^{3/2}}\, .
\end{align*}
Here we made use of
\begin{align*}
  \sum_{i}{}^{'}|v_{i000}|^{2}=\bra{\varphi_{0}} v\ast \varphi_{0}Q
  v\ast \varphi_{0}^{2}\ket{\varphi_{0}} \leq v(0)^{2}\, .
\end{align*}

What these computations show is that
\begin{equation}\label{HNlower}
  H_{N}\geq H_{\text{Bog}}+Nh_{00}+\frac{N}{2}v_{0000} -E_{\eps}\, ,
\end{equation}
where
\begin{align*}
  E_{\eps}&=-\left((1-\eps^{-1})v(0)+v_{0000}\right)\frac{N^{>}(N^{>}-1)}{2(N-1)} + \frac 1{N-1} T_\text{H} + v(0) \frac{\left(N^>\right)^2}{N} \\
  & \quad +\frac{(1-N^{>})^{2}}{\zeta\sqrt{N-1}}+\zeta v(0)^{2}N^{>}\frac{N}{(N-1)^{3/2}}+2v(0)\frac{1+\eps N}{N-1}N^{>}\,  \\
  & \leq
  C\left(\left(\eps^{-1}N^{-1}+\zeta^{-1}N^{-1/2}\right)(N^{>}+1)(T_\text{H}+1)+
    (N^{-1}+\eps+\zeta N^{-1/2})(T_\text{H}+1) \right) \, .
\end{align*}

\subsection{Upper bound} The upper bound on $H_N$ follows essentially
the same lines as the lower bound in the previous subsection. By
Schwarz's inequality
\begin{align*}
  &(P\otimes Q + Q\otimes P )v Q\otimes Q + Q\otimes Q v (P\otimes Q + Q\otimes P ) \\
  &\ \ \leq \eps (P\otimes Q + Q\otimes P )v(P\otimes Q + Q\otimes P
  )+\eps^{-1}Q\otimes QvQ\otimes Q
\end{align*}
and hence
\begin{align*}
  v & \leq P\otimes P v P \otimes P + P \otimes P v Q\otimes Q + Q\otimes Q v P\otimes P \\
  & \quad + (1+\eps)(P\otimes Q + Q\otimes P )v(P\otimes Q + Q\otimes P )+(1+\eps^{-1})v(0)Q\otimes Q \\
  & \quad +P\otimes P v P\otimes Q+ P\otimes P v Q\otimes P + P\otimes
  Q v P\otimes P +Q\otimes P v P \otimes P
\end{align*}
for any $\epsilon>0$.  This means that $H_{N}$ is bounded from above
by the expression (\ref{lowerbound}) with $\eps$ exchanged for
$-\eps$.  Using
\begin{align*}
  \sum_{i,j}{}^{'}\left(h_{ij}-\eps_{0}\delta_{ij}\right)a_{i}^{\dagger}a_{j}= & \sum_{i,j}{}^{'}\left(h_{ij}-\eps_{0}\delta_{ij}\right)b_{i}^{\dagger}b_{j}+\frac{N^{>}-2}{N-1}\sum_{i,j}{}^{'}\left(h_{ij}-\eps_{0}\delta_{ij}\right)a_{i}^{\dagger}a_{j} \\
  \leq &
  \sum_{i,j}{}^{'}\left(h_{ij}-\eps_{0}\delta_{ij}\right)b_{i}^{\dagger}b_{j}+
  \frac{N^{>}T_\text{H}}{N} + v(0) \frac{N^>}{N-1}
\end{align*}
we obtain
\begin{equation}\label{HNupper}
  H_{N}\leq H_{\text{Bog}}+Nh_{00}+\frac{N}{2}v_{0000}+F_{\eps}
\end{equation}
where
\begin{align}\nonumber
  F_{\eps} &= \frac{N^{>}T_\text{H}}{N}+ v(0)\frac{3+2\eps
    N}{N-1}N^{>}+((1+\eps^{-1})v(0)+v_{0000})\frac{N^{>}(N^{>}-1)}{2(N-1)}
  \nonumber\\ \nonumber & \quad +
  \frac{(1-N^{>})^{2}}{\zeta\sqrt{N-1}}+\zeta
  v(0)^{2}N^{>}\frac{N}{(N-1)^{3/2}}\, \\ \nonumber & \leq C\Big(
  \left(N^{-1}+\zeta^{-1}N^{-1/2}+\eps^{-1}N^{-1}
  \right)(N^{>}+1)(T_\text{H}+1) \\ &\ \ \ \ \ \ \ \ +\left(\eps+\zeta
    N^{-1/2}+N^{-1}\right)(N^{>}+1)^{1/2}(T_\text{H}+1)^{1/2} \Big)\,
  . \label{bofe}
\end{align}
Here we have again used that $N^>$ can be bounded by $T_\text{H}$, and
similarly for their square roots. To proceed with the analysis in
Section~\ref{ss:ub}, it is convenient to work with the bound
(\ref{bofe}) on $F_\eps$, involving only the operator
$(N^{>}+1)(T_\text{H}+1)$ and its square root.

\section{Symplectic diagonalization}\label{sympl}

In order to investigate the spectrum of the Bogoliubov Hamiltonian
$H_{\text{Bog}}$ in (\ref{def:hbog}), it is useful to consider first
the usual Bogoliubov Hamiltonian, which is the formal quadratic
expression
\begin{align}\label{us:bog}
  \tilde{H}_{\text{Bog}}&=\frac{1}{2}\left((a^{\dagger})^{\intercal},
    a^{\intercal}\right)\left(\begin{array}{cc}D+V & V \\ V &
      D+V\end{array}\right) \left(\begin{array}{cc} a \\
      a^{\dagger}\end{array}\right) \, .
\end{align}
It is convenient to use a matrix notation where
\begin{align*}
  a=\left(\begin{array}{c}a_{1} \\ a_{2} \\
      \vdots \end{array}\right)\, , \hspace{.5in}
  a^{\dagger}=\left(\begin{array}{c}a_{1}^{\dagger} \\ a_{2}^{\dagger}
      \\ \vdots \end{array}\right)\, ,
\end{align*}
and ${}^{\intercal}$ denotes transposition; e.g., $a^{\intercal}D
a^{\dagger}$ stands for $\sum_{i,j}^{'}\langle \varphi_i| D |
\varphi_j \rangle a_{i}a_{j}^{\dagger}$, $(a^\dagger)^\intercal V a$ stands for $\sum_{i,j}' V_{ij} a_i^\dagger a_j$, etc.  The operator
$\tilde{H}_{\text{Bog}}$ is symmetric since $V$ has real matrix
elements with respect to the basis $\{\varphi_{i}\}_{i\in
  \N}$. Eq.~(\ref{us:bog}) is only a formal expression; in particular,
it has an infinite ground state energy. It also does not preserve the
particle number and hence cannot be restricted to the sector of $N$
particles. Nevertheless, it serves as a useful device to motivate our
analysis below leading to an approximate diagonalization of the
actual Bogoliubov Hamiltonian $H_{\text{Bog}}$.

We introduce the Segal field operators $\phi=(\phi_{1}, \phi_{2},\dots
)^{\intercal}$, $\pi=(\pi_{1},\pi_{2}, \dots )^{\intercal}$, which are
given by
\begin{align*}
  \left(\begin{array}{c}a \\ a^\dagger \end{array}\right)=
  \frac{1}{\sqrt{2}}\left(\begin{array}{cc} 1 & \i \\ 1 &
      -\i \end{array}\right) \left(\begin{array}{c}\phi \\
      \pi \end{array}\right) =: T\left(\begin{array}{c}\phi \\
      \pi \end{array}\right)\, .
\end{align*}
They satisfy the commutation relations
\begin{align*} [\phi_{i},\phi_{j}]=[\pi_{i},\pi_{j}]=0\, ,
  \hspace{.5in} [\phi_{i},\pi_{j}]=i\delta_{ij}\, .
\end{align*}
These remain invariant under symplectic transformations $S$, which
satisfy
\begin{align*}
  S^\intercal JS=J= \left(\begin{array}{cc} 0 & 1 \\ -1 &
      0 \end{array} \right)\, .
\end{align*}

We can write
\begin{align*}
  \tilde{H}_{\text{Bog}}=\left(\phi^{\intercal}, \pi^{\intercal}
  \right) M\left(\begin{array}{c}\phi \\ \pi \end{array}\right)
\end{align*}
where
\begin{align*}
  M:=\frac 1 2 \,T^{*} \left(\begin{array}{cc}D+V & V \\ V &
      D+V\end{array}\right)T= \frac 1 2 \left(\begin{array}{cc} D+2V &
      0 \\ 0 & D \end{array} \right)\, .
\end{align*}
Here and in the following, we shall use $^*$ for the adjoint of an
operator on the one-particle space $\mathcal{F}^{(1)}$ or the doubled
space $\mathcal{F}^{(1)}\oplus \mathcal{F}^{(1)}$, while we use
$^\dagger$ for the adjoint of an operator on Fock space.

In order to diagonalize $\tilde{H}_{\text{Bog}}$ we thus have to
symplectically diagonalize $M$. To do so we introduce a real unitary
operator $U_{0}$ such that
\begin{align*}
  \hat{E}=U_{0}^{*}EU_{0}
\end{align*}
is diagonal with ordered eigenvalues,
i.e. $\hat{E}=\sum_{i}'e_{i}\ket{\varphi_{i}}\bra{\varphi_{i}}$ with
$0<e_{1}\leq e_{2}\leq \dots$. On the subspace $Q L^{2}(\R^{d})$, the
operators $D$, $E$ and $\hat E$ are invertible, and we denote their
inverse by $D^{-1}$, $E^{-1}$ and $\hat E^{-1}$ for simplicity, i.e.,
$D^{-1} = Q (QD)^{-1}$, etc.

With
\begin{equation}\label{def:S}
  S=\left(\begin{array}{cc} D^{{1/2}} & 0 \\ 0 & D^{-{1/2}} \end{array} \right)\left(\begin{array}{cc} U_{0} & 0 \\ 0 & U_{0} \end{array} \right)\left(\begin{array}{cc} \hat{E}^{-{1/2}} & 0 \\ 0 & \hat{E}^{{1/2}} \end{array} \right)=\left(\begin{array}{cc} AU_{0} & 0 \\ 0 & BU_{0} \end{array} \right) \,,
\end{equation}
where $A:=D^{1/2}E^{-1/2}$ and $B:=(A^{-1})^{*}$, we then have
$S^\intercal = S^*$ and
\begin{align*}
  S^{{*}}MS= \frac 1 2 \left(\begin{array}{cc} \hat{E} & 0 \\ 0 &
      \hat{E} \end{array} \right)\, .
\end{align*}
This corresponds to a Hamiltonian consisting of sums of independent
harmonic oscillators of the form $\phi_i^2 + \pi_i^2$, and hence
yields the desired diagonalization of $\tilde{H}_{\text{Bog}}$.

\begin{rem}\label{rem:bdg}
  As claimed in Remark~\ref{nonrig} in Section~\ref{ss:mod}, it is not
  difficult to see that the positive eigenvalues $\omega$ in
  (\ref{physeigprob}) are precisely the eigenvalues of $E$. With
  \begin{align*}
    I:=\left(\begin{array}{cc}-\i& 0 \\ 0 & \i \end{array} \right)
  \end{align*}
  Eq.~(\ref{physeigprob}) can be written as $$ 2 \i I T M T^* \psi =
  \omega \psi\,,
$$
where we denote $\psi = (u,y)^\intercal$ for short.  If we multiply
this from the left with $I TS^{*}JT^{*}$, using $T^* I = J T^*$ and
$S^* J = J S^{-1}$, we obtain the equation
$$
2 \i I T S^* M S T^* \chi = \omega \chi \,,
$$
with $\chi = TS^{-1}T^{*} \psi$. This latter equation is simply
\begin{align*}
  \left(\begin{array}{cc}\hat{E}& 0 \\ 0 & -\hat{E} \end{array}
  \right) \chi= \omega \chi\,,
\end{align*}
hence $\omega$ is indeed an eigenvalue of $\hat E$, as claimed.
\end{rem}

The formal considerations above serve as a starting point of our
analysis. Using $S$ in (\ref{def:S}), we define particle number
preserving operators $c=(c_{1},c_{2},\dots)$ by
\begin{align}
  \left(\begin{array}{cc} b \\ b^{\dagger} \end{array} \right) &=
  \frac 1 2 \left(\begin{array}{cc} 1 & \i \\ 1 &
      -\i \end{array}\right) \left(\begin{array}{cc}AU_{0} & 0 \\ 0 &
      BU_{0} \end{array}\right)
  \left(\begin{array}{cc} 1 & 1 \\ -\i & \i \end{array}\right) \left(\begin{array}{c}  {c} \\ {c}^{\dagger} \end{array}\right) \nonumber \\
  &= \frac 1 2 \left(\begin{array}{cc} AU_{0}+BU_{0}& AU_{0}-BU_{0} \\
      AU_{0}-BU_{0} & AU_{0}+BU_{0} \end{array}\right)
  \left(\begin{array}{c} {c} \\ {c}^{\dagger} \end{array}\right)\,
  . \label{transformation}
\end{align}
Note that the operators $A$, $B$ and $U_0$ are all real, hence
$c^\dagger_j$ is indeed the adjoint of $c_j$.  By inverting $S$ one
easily obtains the inverse transformation law
\begin{align}\label{invtrafo}
  \left(\begin{array}{c} {c} \\ {c}^{\dagger} \end{array}\right)=\frac
  1 2 \left(\begin{array}{cc}U_{0}^{*} A^{-1}+U_{0}^{*}B^{-1}&
      U_{0}^{*}A^{-1}-U_{0}^{*}B^{-1} \\
      U_{0}^{*}A^{-1}-U_{0}^{*}B^{-1} &
      U_{0}^{*}A^{-1}+U_{0}^{*}B^{-1} \end{array}\right)
  \left(\begin{array}{c} {b} \\ {b}^{\dagger} \end{array}\right) \, .
\end{align}
We can rewrite the Bogoliubov Hamiltonian
\begin{align}\label{bogomatrix}
  H_{\text{Bog}} = (b^{\dagger})^{{\intercal}}(D+ V) b+\frac 1 2
  b^{{\intercal}}Vb+\frac 1 2 (b^{\dagger})^{{\intercal}} V
  b^{\dagger}
\end{align}
as a quadratic operator in these $c$, $c^{\dagger}$. Here,
$(b^\dagger)^\intercal V b = \sum_{i,j}' V_{ij} b_i^\dagger b_j$,
etc. We insert (\ref{transformation}) into (\ref{bogomatrix}) and
obtain {\begin{align}\label{BogoHam} \nonumber
    H_{\text{Bog}} & =  {\sum_{i}{}^{'} e_{i}{c}_{i}^{\dagger}{c}_{i}}- \sum_{i,j}{}^{'}  \left(U_{0}^{{*}}\left( Y-\frac{E}{2}\right)U_{0}\right)_{ij}[{c}_{i},{c}^{\dagger}_{j}]  \\
    & \quad -{\frac 1 2
      \sum_{i,j}{}^{'}Z_{ij}\left([{c}_{j},{c}_{i}]+[{c}_{i}^{\dagger},{c}_{j}^{\dagger}]
      \right)}=: \text{(I)+(II)+(III)} \, ,
  \end{align}} where
$$
Y:=\frac 1 4
E^{1/2}D^{-1/2}\left({D}+{V}\right)D^{1/2}E^{-1/2}+\text{h.c.}
$$
and
\begin{align*}
  {Z:}&=\frac{1}{4}U_{0}^{*}\left[(A-B)^{*}(D+V)(A+B)+\frac 1 2 (A+B)^{*} V (A+B)+ \frac 1 2 (A-B)^{*}V (A-B)\right]U_{0} \\
  &= \frac 1 4 U_{0}^{*}\left[{A^{*}(D+2V)A-B^{*}DB}{-B^{*}(D+V)A+A^{{*}}(D+V)B}\right] U_{0} \\
  &= \frac 1 4 U_{0}^{{*}}\left[A^{{*}}(D+V)B -B^{{*}}(D+V)A\right]
  U_{0} \,.
\end{align*}
Note that $Z$ is antisymmetric and hence
\begin{align*} \sum_{i,j}{}^{'} Z_{ij}{c}_{i} {c}_{j} &= \frac 1 2
  \sum_{i,j}{}^{'} Z_{ij}{c}_{i}{c}_{j}- \frac 1 2 \sum_{i,j}{}^{'}
  Z_{ji}{c}_{i}{c}_{j} = - \frac 1 2 \sum_{i,j}{}^{'}
  Z_{ji}[{c}_{i},{c}_{j}] \,.
\end{align*}

To arrive at (\ref{BogoHam}), we have used that
\begin{align*}
  {}&\frac{1}{4} \left[(A+B)^{{*}}(D+V)(A+B)+\frac 1 2 (A-B)^{{*}}V(A+B)+\frac 1 2  (A+B)^{{*}}V(A-B) \right] \\
  &= \frac{1}{4} \left[(A+B)^{{*}}(D+V)(A+B)+A^{{*}}VA-B^{{*}}VB \right] \\
  &= \frac{1}{4} \left[  A^{{*}}(D+2V)A +B^{{*}}DB+B^{{*}}(D+V)A+A^{{*}}(D+V)B \right] \\
  &= \frac 1 2 E +\frac{1}{4}\left[{B^{{*}}(D+V)A+A^{{*}}(D+V)B} \right] \\
  &=\frac 1 2 E + Y
\end{align*}
and
\begin{align*}
  {}&\frac{1}{4}\left[(A-B)^{{*}}(D+V)(A-B)+\frac{1}{2}(A+B)^{{*}}V(A-B)+\frac{1}{2} (A-B)^{{*}}V(A+B) \right] \\
  &=  \frac{1}{4} \left[(A-B)^{{*}}(D+V)(A-B)+A^{{*}}VA-B^{{*}}VB \right] \\ &= \frac{1}{4}  \left[A^{{*}}(D+2V)A+B^{{*}}D B -B^{{*}}(D+V)A-A^{{*}}(D+V)B\right] \\
  &= \frac 12 E -\frac{1}{4}\left[B^{{*}}(D+V)A+A^{{*}}(D+V)B \right] \\
  &= \frac 12 E -Y\, .
\end{align*}

\section{Bounds on the Bogoliubov Hamiltonian} \label{sec_est}

To prove Theorem~\ref{mainthm} we derive upper and lower bounds for
the various terms (I)--(III) in (\ref{BogoHam}). This yields a bound
on $H_{\mathrm{Bog}}$ in terms of an operator whose spectrum is
explicit, as well as errors which are small for large $N$ in the low-energy sector. More
specifically we shall prove:

\begin{prop}\label{maintechprop} The three terms in (\ref{BogoHam})
  have the following properties. There exists a unitary operator
  $\mathcal{U}:\mathcal{F}^{(N)}\to\mathcal{F}^{(N)}$ (explicitly
  given in (\ref{defofU}) below) such that the following bounds hold
  on $\mathcal{F}^{(N)}$:
\begin{itemize} 
\item[(I):] For arbitrary $\lambda>0$ we have 
\begin{align*}
\sum_{i}{}^{'}e_{i}c_{i}^{\dagger}c_{i}& \geq  (1-\lambda)\mathcal{U}^{\dagger}\left(\sum_{i}{}^{'}e_{i} a_{i}^{\dagger}a_{i}\right)\mathcal{U}-C(1+\lambda^{-1})N^{-1}(N^{>}+1)(T_\text{H}+1) \, ,  \\
\sum_{i}{}^{'}e_{i}c_{i}^{\dagger}c_{i} &\leq  (1+\lambda)\mathcal{U}^{\dagger}\left(\sum_{i}{}^{'}e_{i} a_{i}^{\dagger}a_{i}\right)\mathcal{U}+C(1+\lambda^{-1})N^{-1}(N^{>}+1)(T_\text{H}+1)\, .
\end{align*}
\item[(II):] $D+V-E$ is a trace class operator, and
\begin{align*}
& -CN^{-1}(T_\text{H}+1) \\ & \qquad \leq 2 \sum_{i,j}{}^{'}  \Big( U_{0}^{{*}}  \left( Y-\frac{E}{2}\right)U_{0}\Big) _{ij}[{c}_{i},{c}^{\dagger}_{j}] -\mathrm{tr}(D+V-E) + v_{0000} \\ & \qquad \qquad \leq CN^{-1}(T_\text{H}+1)\, .
\end{align*}
\item[(III):] \begin{align*}
-CN^{-1}T_\text{H}\leq\sum_{i,j}{}^{'}Z_{ij}\left([{c}_{j},{c}_{i}]+[{c}_{i}^{\dagger},{c}_{j}^{\dagger}] \right) \leq CN^{-1}T_\text{H}\, .
\end{align*}
\end{itemize}

\end{prop}

The following three subsections contain the proof of this proposition.

\subsection{Proof of Proposition~\ref{maintechprop} (I)} \label{ss:I}

We first refine the symplectic transformation $S$. By polar
decomposition there is a unitary $W_{0}$ such that
$A=|A^{*}|W_{0}=W_{0}|A|$. Since
\begin{align}\label{lucky}
  |B^{*}|=|A^{-1}|=|A^{*}|^{-1}
\end{align}
also $B=|B^{*}|W_{0}$. Hence
\begin{align*}
  S&=\left(\begin{array}{cc} |A^{{*}}|W_{0}U_{0} & 0 \\ 0 &
      |B^{{*}}|W_{0}U_{0} \end{array} \right) =: \tilde{S}
  \left(\begin{array}{cc} W & 0 \\ 0 & W \end{array} \right) \, ,
\end{align*}
where $W=W_{0}U_{0}$.  The transformation
\begin{align*}
  \left(\begin{array}{cc} W & 0 \\ 0 & W \end{array} \right)
\end{align*}
is implementable on $\mathcal{F}$ by a unitary
$\mathcal{W}=\Gamma(W)$, as it corresponds to a change of basis of the
one-particle Hilbert space $L^{2}(\R^{d})$.  We define the real,
bounded, and positive operator
\begin{align*}
  \alpha&:= \log \left(|A^{{*}}|^{-1}\right)\, .
\end{align*}
Note that $\log |B^{*}|=\alpha$ due to (\ref{lucky}). One can show
that for any $t\in \R$ the symplectic transformation
\begin{align*}
  \tilde{S}_{t}:=\left(\begin{array}{cc} \e^{-t\alpha} & 0 \\ 0 &
      \e^{t\alpha} \end{array} \right)
\end{align*}
is implemented on Fock space $\mathcal{F}$ by $\e^{X_{a}t}$ where
\begin{align*}
  X_{a}&:=\frac 1 2 \sum_{i,j}{}^{'}
  \alpha_{ij}(a_{i}^{\dagger}a_{j}^{\dagger}-a_{i}a_{j}) \, .
\end{align*}
However, it is important to note that $\e^{X_{a}t}$ does not preserve
the particle number, and hence we shall instead work with $\e^{Xt}$,
where
\begin{align*}
  X:= \frac 1 2 \sum_{i,j}{}^{'}
  \alpha_{ij}(b_{i}^{\dagger}b_{j}^{\dagger}-b_{i}b_{j}) \, .
\end{align*}
We will repeatedly need the following facts.

\begin{lem} \label{hscondi}
  \begin{enumerate}
  \item $V$ is a positive trace class operator.
  \item $A-1$ and $B-1$ are Hilbert-Schmidt operators.
  \item $\alpha$ is a Hilbert-Schmidt operator.
  \item $iX:\mathcal{F}^{(N)}\to \mathcal{F}^{(N)}$ is a symmetric
    bounded operator.
  \end{enumerate}
\end{lem}

The proof of this lemma will be given at the end of this subsection.
The lemma implies that
\begin{align}\label{defofU}
  \mathcal{U}:=\mathcal{W}^{\dagger}\e^{X}
\end{align}
is a particle number preserving unitary transformation on the Fock
space $\mathcal{F}$, and hence we can study its restriction to the
$N$-particle sector $\mathcal{F}^{(N)}$. With
$\nu:=\frac{a_{0}^{2}}{N-1}$, $G:=\cosh(\alpha)W^{}$, and
$H:=\sinh(\alpha)W^{}$, we define the operators $d_i:
\mathcal{F}^{(N)}\to \mathcal{F}^{(N-1)}$ by
\begin{align}\label{def:di}
  d_{i}:=a(G\varphi_{i})+\nu a^{\dagger}(H\varphi_{i})\, .
\end{align}
Note that (\ref{invtrafo}) implies
\begin{align*}
  c_{i}&= \frac{1}{\sqrt{N-1}}\sum_{j}{}^{'}\left(\left( W^{*} \cosh(\alpha) \right)_{ij}a_{j}a_{0}^{\dagger}+\left( W^{*} \sinh(\alpha) \right)_{ij}a_{j}^{\dagger}a_{0}\right) \\
  &=\frac{a(G\varphi_{i})a_{0}^{\dagger}}{\sqrt{N-1}}+\frac{a(H\varphi_{i})a_{0}^{\dagger}}{\sqrt{N-1}}
\end{align*}
from which we derive
\begin{align} \nonumber
  d_{i}^{\dagger}d_{i}&=c_{i}^{\dagger}c_{i}+a^{\dagger}(G\varphi_{i})a(G\varphi_{i})\left(1-\frac{a_{0}a_{0}^{\dagger}}{N-1}\right)
  \\&\ \ \ \ \ \
  +a(H\varphi_{i})a^{\dagger}(H\varphi_{i})\left(\frac{(a_{0}^{\dagger})^{2}a_{0}^{2}}{(N-1)^{2}}-\frac{a_{0}^{\dagger}a_{0}}{N-1}\right)
  \nonumber \\ \nonumber
  &=c_{i}^{\dagger}c_{i}+a^{\dagger}(G\varphi_{i})a(G\varphi_{i})\frac{N^{>}-2}{N-1} \\
  & \ \ \ \ \ \
  -(a^{\dagger}(H\varphi_{i})a(H\varphi_{i})+\|H\varphi_{i}\|^{2})\frac{N^{>}(N-N^{>})}{(N-1)^{2}} \label{cdrel}\,
  .
\end{align}
The first step towards the proof of Proposition~\ref{maintechprop} (I)
is the following lemma.
\begin{lem} \label{firststeplem}
  \begin{align*}
    \sum_{i}{}^{'}e_{i}d_{i}^{\dagger}d_{i}-CN^{-1}T_\text{H}N^{>}\leq
    \sum_{i}{}^{'}e_{i}c_{i}^{\dagger}c_{i}\leq
    \sum_{i}{}^{'}e_{i}d_{i}^{\dagger}d_{i}+CN^{-1}(N^{>}+1)(T_\text{H}+1)
  \end{align*}
\end{lem}
\begin{proof}
  By (\ref{cdrel}) we have
  \begin{align}\label{lowerboundspec}
    \sum_{i}{}^{'}e_{i}c_{i}^{\dagger}c_{i}\geq
    \sum_{i}{}^{'}e_{i}d_{i}^{\dagger}d_{i}-\frac{N^{>}}{N}\sum_{i}{}^{'}e_{i}a^{\dagger}(G\varphi_{i})a(G\varphi_{i})\,
  \end{align}
  and similarly
  \begin{align}\label{upperboundspec}
    \sum_{i}{}^{'}e_{i}c_{i}^{\dagger}c_{i}\leq& \sum_{i}{}^{'}e_{i}d_{i}^{\dagger}d_{i}\nonumber \\ &+\sum_{i}{}^{'}e_{i}\Big(\frac{2}{N-1}a^{\dagger}(G\varphi_{i})a(G\varphi_{i})\nonumber \\
    & \ \ \ \ \ \ \ \ \ \ \ \ \ \ \ \ \
    +\frac{NN^{>}}{(N-1)^{2}}\left(a^{\dagger}(H\varphi_{i})a(H\varphi_{i})+\|H\varphi_{i}\|^{2}\right)
    \Big)\, .
  \end{align}
  Hence the lemma follows from the following three estimates:
  \begin{align}\nonumber
    \sum_{i}{}^{'}e_{i}a^{\dagger}(G\varphi_{i})a(G\varphi_{i}) &\leq
    C T_\text{H}\, , \\ \nonumber
    \sum_{i}{}^{'}e_{i}a^{\dagger}(H\varphi_{i})a(H\varphi_{i}) &\leq
    C N^{>}\, , \\ \label{tes} \sum_{i}{}^{'}e_{i}
    \|H\varphi_{i}\|^{2}&<\infty\, .
  \end{align}

  For the proof of the first estimate note that the operator on the
  left side is the second quantization of $ G U_0^{{*}}E U_0 G^{{*}}$,
  which we can write as
  \begin{align*}
    G U_0^{{*}}E U_0 G^{{*}} &= \cosh(\alpha)W_{0}EW_{0}^{{*}}\cosh(\alpha) \\
    &= \frac 1 4 \left(D^{-1/2}E^{1/2}+D^{1/2}E^{-1/2} \right)E\left(E^{1/2}D^{-1/2}+E^{-1/2}D^{1/2} \right) \\
    &=\frac 1 4
    D^{1/2}\left(D^{-1}E^{1/2}+E^{-1/2}\right)E\left(E^{1/2}D^{-1}+E^{-1/2}\right)D^{1/2}\,.
  \end{align*}
  Since $T_\text{H}=\mathrm{d}\Gamma(D)$ it suffices to show
  boundedness of the operator
  \begin{align*}
    D^{-1}E^{2}D^{-1}+D^{-1}E+ED^{-1}+1\, ,
  \end{align*}
  which follows from
  \begin{align*}
    \|D^{-1}E^{2}D^{-1}\|= \|1 + 2 D^{-1/2}VD^{-1/2}\|<\infty\, .
  \end{align*}
  The second estimate in (\ref{tes}) follows from the third, for which
  we note that
  \begin{align*}
    \sum_{i}{}^{'}e_{i} \|H\varphi_{i}\|^{2}&=\mathrm{tr} (HU_{0}^{{*}}EU_{0} H^{{*}})\\
    &=\|E^{1/2}W_{0}^{{*}}\sinh(\alpha)\|_{2}^{2} \\
    &=\frac 1 4 \|(D-E)D^{-1/2} \|_{2}^{2}\, .
  \end{align*}
  To bound the Hilbert-Schmidt norm of the operator $(D-E)D^{-1/2}$,
  we use the integral representation $x^{1/2}=\pi
  \int_{0}^{\infty}t^{1/2}\left(\frac{1}{t}-\frac{1}{x+t}\right)dt$,
  which implies that
  \begin{align} \nonumber \|(D-E)D^{-1/2}\|_{2}&= \pi \left\|
      \int_{0}^{\infty}t^{1/2}\left( (t+D^{2})^{-1}-(t+E^{2})^{-1}
      \right)D^{-1/2}dt \right\|_{2} \\ \nonumber & = 2\pi \left\|
      \int_{0}^{\infty}t^{1/2} (t+E^{2})^{-1}D^{1/2}V(t+D^{2})^{-1} dt
    \right\|_{2} \\ \label{Hilbertschmidtest} &\leq 2\pi \|V\|_{2}
    \int_{0}^{\infty}t^{1/2}\|(t+E^{2})^{-1}D^{1/2}\|
    \|(t+D^{2})^{-1}\| dt\, .
  \end{align}
  Using $D\leq E$ and the spectral theorem one verifies that
  $\|(t+E^{2})^{-1}D^{1/2}\| \leq \|(t+E^{2})^{-1}E^{1/2}\|\leq
  C(1+t^{3/4})^{-1}$. Since also $\|(t+D^{2})^{-1}\| \leq 1/t$, the
  integrand in the last line in (\ref{Hilbertschmidtest}) falls off
  like $t^{-5/4}$ at infinity, making the integral finite.
\end{proof}
Proposition~\ref{maintechprop} (I) is now a direct consequence of the
following lemma.

\begin{lem} \label{secondsteplem} For arbitrary $\lambda>0$ we have
  the bounds
  \begin{align*}
    \sum_{i}{}^{'}e_{i}d_{i}^{\dagger}d_{i}& \geq  (1-\lambda)\mathcal{U}^{\dagger}\left(\sum_{i}{}^{'}e_{i} a_{i}^{\dagger}a_{i}\right)\mathcal{U}-C\lambda^{-1}N^{-1}(N^{>}+1)^{2} \, ,  \\
    \sum_{i}{}^{'}e_{i}d_{i}^{\dagger}d_{i} &\leq
    (1+\lambda)\mathcal{U}^{\dagger}\left(\sum_{i}{}^{'}e_{i}
      a_{i}^{\dagger}a_{i}\right)\mathcal{U}+C(1+\lambda^{-1})N^{-1}(N^{>}+1)^{2}\,
    .
  \end{align*}
\end{lem}

\begin{proof}
  We define $K_{i}:= \mathcal{U}^{\dagger}a_{i}\mathcal{U}-d_{i}$ and
  hence, by Schwarz's inequality,
  \begin{align*}
    (1-\lambda)\mathcal{U}^{\dagger}a_{i}^{\dagger}a_{i}\mathcal{U}-\lambda^{-1}K_{i}^{\dagger}K_{i}
    \leq d_{i}^{\dagger}d_{i} \leq
    (1+\lambda)\mathcal{U}^{\dagger}a_{i}^{\dagger}a_{i}\mathcal{U}+(1+\lambda^{-1})K_{i}^{\dagger}K_{i}
  \end{align*}
  for arbitrary $\lambda>0$.  We have to show that
  $\sum_{i}{}^{'}e_{i}K_{i}^{\dagger}K_{i} \leq CN^{-1}(N^{>}+1)^{2}$.
  To simplify notation we define also $g_{t}:=\cosh(\alpha t)$,
  $h_{t}:=\sinh(\alpha t)$ and the quantity
  \begin{align*}
    \kappa_{f}(t):= \e^{-t X}a(f)\e^{tX}-a(g_{t}f)-\nu
    a^{\dagger}(h_{t}\overline{f}) \quad \text{for} \ f \perp
    \varphi_{0} \,,
  \end{align*}
  which is related to $K_{i}$ by $K_{i}=\kappa_{W\varphi_{i}}(1)$. We
  claim
  \begin{align}
    \label{kappaest}
    \kappa_{f}(1)^{\dagger}\kappa_{f}(1)&\leq
    CN^{-1}(N^{>}+1)^{2}\langle f| \alpha^{2}| f\rangle \, .
  \end{align}
  Assuming this for the moment, we can use $\hat{E}=U_{0}^{*}EU_{0}$
  as well as $W_{0}^{*}\alpha^{2}W_{0}=(\log |A|)^{2}$ to conclude
  that
  \begin{align*}
    \sum_{i}{}^{'}e_{i}K_{i}^{\dagger}K_{i} &\leq CN^{-1}(N^{>}+1)^{2} \sum_{i}{}^{'}e_{i}\langle W\varphi_{i}|\alpha^{2}W\varphi_{i}\rangle \\
    &= CN^{-1}{(N^{>}+1)^{2}}\sum_{i}{}^{'}\langle E^{1/2}(\log
    |A|)^{2}E^{1/2}U_{0}\varphi_{i}| U_{0}\varphi_{i }\rangle \\
    &=CN^{-1}{(N^{>}+1)^{2}} \| E^{1/2}(\log |A|)^{2}E^{1/2}\|_{1} \,.
  \end{align*}
  The claim of the lemma then follows if $\|E^{1/2} (\log
  |A|)^{2}E^{1/2} \|_{1}<\infty $. To see this observe that
  \begin{align}\label{logestA}
    0\leq  -\log|A| &\leq |A|^{-1}-1 \nonumber \\
    & \leq |A|^{-2}-1= E^{1/2}D^{-1}E^{1/2}-1 \,
  \end{align}
  which leads to
  \begin{align*}
    \| E^{1/2}(\log |A|)^{2}E^{1/2}\|_{1} &\leq  \|E^{1/2}(E^{1/2}D^{-1}E^{1/2}-1)^{2}E^{1/2}\|_{1}\\
    &= \|(ED^{-1}-1)E^{1/2}\|_{2}^{2} \\
    &= \|(E-D)D^{-1}E^{1/2}\|_{2}^{2} \\
    &\leq \| (E-D)D^{-1/2}\|_{2}^{2}\|D^{-1/2}E^{1/2}\|^{2}\,.
  \end{align*}
  The claim thus follows from (\ref{Hilbertschmidtest}) and
  boundedness of $B=D^{-1/2}E^{1/2}$.

  The proof of (\ref{kappaest}) is a bit more elaborate. With
  \begin{align*} [X,a(f)]= -\nu a^{\dagger}(\alpha \overline{f})
  \end{align*}
  we easily obtain
  \begin{align*}
    \kappa_{f}''(t)=\kappa_{\alpha^{2}f}(t)-r_{\alpha,\alpha f}(t)
  \end{align*}
  where
  \begin{align*}
    r_{\alpha, \varphi}(t):=\e^{-tX} \left[(1-\nu
      \nu^{\dagger})a(\alpha
      \varphi)-[X,\nu]a^{\dagger}(\overline{\varphi})\right]\e^{tX}\,
    .
  \end{align*}
  Using $\kappa_{f}(0)=\kappa'_{f}(0)=0$, a second order Taylor
  expansion yields
  \begin{align*}
    \kappa_{f}(t)= \int_{0}^{t}(t-s)\left(\kappa_{\alpha^{2}f}(s)
      -r_{\alpha, f}(s)\right) ds\, .
  \end{align*}
  For any $\ket{\psi} \in \mathcal{F}^{(N)}$ we introduce
  \begin{align*}
    \hat{\kappa}_{\psi}(t)&:= \sup_{\|\alpha f\|\leq1} \|\kappa_{f}(t) \ket{\psi} \| \, ,\\
    \hat{r}_{\alpha,\psi}&:= \frac 1 2\sup_{s\leq 1} \sup_{\|\alpha
      f\| \leq 1} \|r_{\alpha, f}(s)\ket{\psi}\| \, .
  \end{align*}
  Note that
  \begin{align*}
    \sup_{\|\alpha f\|\leq1} \|\kappa_{\alpha^{2}f}(t) \ket{\psi} \|
    =\|\alpha\|^{2} \sup_{\|\alpha f\|\leq1}
    \|\kappa_{\alpha^{2}f/\|\alpha\|^{2}}(t) \ket{\psi} \| \leq
    \|\alpha\|^{2} \hat{\kappa}_{\psi}(t)
  \end{align*}
  which yields
  \begin{align*}
    \hat{\kappa}_{\psi}(t)\leq \hat{r}_{\alpha,\psi}+ \|\alpha\|^{2}
    \int_{0}^{t} \hat{\kappa}_{\psi}(s)ds
  \end{align*}
  for $t\leq 1$. It follows from Gr\"onwall's lemma (see, e.g.,
  \cite[Thm. III.1.1]{gron}) that
  \begin{align*}
    \hat{\kappa}_{\psi}(1) \leq \e^{\|\alpha\|^{2}}
    \hat{r}_{\alpha,\psi}(1)\, .
  \end{align*}
  If $f \in \ker \alpha$ then $\kappa_{f}(t)=0$. For $f \notin \ker
  \alpha$
  \begin{align*}
    \frac{\|\kappa_{f}(1)\ket{\psi} \|}{\|\alpha f\|}\leq
    \e^{\|\alpha\|^{2}} \hat{r}_{\alpha,\psi}(1)
  \end{align*}
  from which (\ref{kappaest}) follows if we can show that
  \begin{align}\label{remest}
    \hat{r}_{\alpha,\psi}(1) \leq C N^{-1/2}\| (N^{>}+1)\ket{\psi}
    \|\, .
  \end{align}

  To see (\ref{remest}) we define $g=\alpha f$ and first show that for
  $\|g\|\leq1$
  \begin{align} \label{opest1}
    a^{\dagger}(\alpha\overline{g})(1-\nu\nu^{\dagger})^{2}a(\alpha\overline{g})
    &\leq CN^{-1}{(N^{>}+1)^{2}}
  \end{align}
  and
  \begin{align}
    a(g)[X,\nu]^{\dagger}[X,\nu]a^{\dagger}(g)&\leq
    CN^{-1}{(N^{>}+1)^{2}} \, . \label{opest2}
  \end{align}
  The first bound follows directly from
  \begin{align*}
    1-\nu\nu^{\dagger}&=
    \frac{(2N+3)N^{>}-(N^{>})^{2}-5N-1}{(N-1)^{2}}
  \end{align*}
  and $a^{\dagger}(\alpha \overline{g}) a(\alpha \overline{g}) \leq
  \|\alpha\|^{2} N^{>}$. To show (\ref{opest2}) we write
  \begin{align*}
    a(g)[X,\nu]^{\dagger}[X,\nu]a^{\dagger}(g)&\leq \left(\frac{[(a_{0}^{\dagger})^{2},a_{0}^{2}]}{2(N-1)^{2}}\right)^{2}    a(g)\left(\sum_{i,j,k,l}{}^{'}\alpha_{ij} \alpha_{kl}a_{k}^{\dagger}a_{l}^{\dagger}a_{i}a_{j}\right) a^{\dagger}(g) \\
    &\leq \left(\frac{2(N-N^{>})+1}{(N-1)^{2}}\right)^{2}    a(g)\|\alpha\|_{2}^{2}N^{>}(N^{>}-1)a^{\dagger}(g) \\
    &\leq CN^{-2}\|\alpha\|_{2}^{2}\|g\|^{2}N^{>}(N^{>}-1)(N^{>}+1) \\
    &\leq C N^{-1}(N^{>}+1)^{2}\, .
  \end{align*}
 
  To conclude the proof of (\ref{kappaest}) it remains to show that
  the inequality
  \begin{align}\label{opest3}
    \e^{-tX}(N^{>}+1)^{2}\e^{tX}& \leq \e^{Ct}(N^{>}+1)^{2}
  \end{align}
  holds for $t=1$.  To that end we compute
  \begin{align*}
    [X,N^{>}]&=-\frac 1 2 \sum_{i,j}{}^{'}\left(\nu^{\dagger}\alpha_{ij}[a_{i}a_{j},N^{>}]-\nu \alpha_{ij}[a_{i}^{\dagger}a_{j}^{\dagger},N^{>}] \right) \\
    &=
    -\sum_{i,j}{}^{'}\alpha_{ij}(b_{i}b_{j}+b_{i}^{\dagger}b_{j}^{\dagger})
    \, .
  \end{align*}
  Taking the square of this expression yields
  \begin{align} \nonumber
    [X,N^{>}]^{2} &\leq 2 \left(\sum_{i,j}{}^{'} \alpha_{ij}b_{i}b_{j}\right)\left(\sum_{k,l}{}^{'} \alpha_{kl}b_{k}^{\dagger}b_{l}^{\dagger}\right) +2\left(\sum_{i,j}{}^{'} \alpha_{ij}b_{i}^{\dagger}b_{j}^{\dagger} \right) \left(\sum_{k,l}{}^{'}\alpha_{kl}b_{k}b_{l} \right) \\
    &\leq 2
    \nu^{\dagger}\nu\sum_{i,j,k,l}{}^{'}\alpha_{ij}\alpha_{kl}a_{i}a_{j}a_{k}^{\dagger}a_{l}^{\dagger}+
    2\|\alpha\|_{2}^{2} \nu \nu^{\dagger}N^{>}(N^{>}-1) \nonumber \\
    \nonumber &\leq 2
    \frac{(N-N^{>})(N-N^{>}-1)}{(N-1)^{2}}\left(\sum_{i,j,k,l}{}^{'}\alpha_{ij}\alpha_{kl}a_{i}^{\dagger}a_{j}^{\dagger}a_{k}a_{l}+4\sum_{i,j}{}^{'}(\alpha^{2})_{ij}a_{i}^{\dagger}a_{j}
      +2 \|\alpha\|_{2}^{2} \right) \\ \nonumber
    &\ \ \ \ + 2 \|\alpha\|_{2}^{2}\frac{(N-N^{>})(N-N^{>}+3)+2}{(N-1)^{2}}N^{>}(N^{>}-1) \\
    &\leq 2 \frac{(N-N^{>})(N-N^{>}-1)}{(N-1)^{2}}\left(
      \|\alpha\|_{2}^{2} N^{>}(N^{>}-1) + 4
      \|\alpha\|^{2}N^{>}+2\|\alpha\|_{2}^{2}\right) \nonumber \\
    \nonumber
    &\ \ \ \ +2 \|\alpha\|_{2}^{2}\frac{(N-N^{>})(N-N^{>}+3)+2}{(N-1)^{2}}N^{>}(N^{>}-1) \\
    &\leq C\|\alpha\|^{2}_{2}(N^{>}+1)^{2} \, . \label{squarecommest}
  \end{align}
  By Schwarz's inequality we obtain
  \begin{align*}
    [X,(N^{>}+1)^{2}]&=(N^{>}+1)[X,N^{>}]+[X,N^{>}](N^{>}+1) \\
    &\leq C(N^{>}+1)^{2}\,
  \end{align*}
  and hence it follows that
  \begin{align*}
    \e^{tX}(N^{>}+1)^{2}\e^{-tX} &= (N^{>}+1)^{2}+\int_{0}^{t}\e^{sX}[X,(N^{>}+1)^{2}]\e^{-sX}ds \\
    & \leq (N^{>}+1)^{2}+C\int_{0}^{t}\e^{sX}(N^{>}+1)^{2}\e^{-sX}ds\,
    .
  \end{align*}
  Gr\"onwall's lemma then yields (\ref{opest3}). This completes the
  proof of the lemma.
\end{proof}

We conclude this section with the proof of Lemma~\ref{hscondi}.
\begin{proof}[Proof of Lemma~\ref{hscondi}]
  (i): The positivity of $V$ follows directly from the assumption that
  $v$ is of positive type:
  \begin{align*}
    \bra{\psi}V\ket{\psi}=\int_{\R^{2d}}
    \varphi_{0}(x)\varphi_{0}(y)v(x-y)\overline{\psi(x)}\psi(y)dxdy
    \geq 0\, .
  \end{align*}
  In particular, the trace norm of $V$ equals its trace, which is
  equal to
$$
\Tr\, V = \int_{\R^{d}}\varphi_{0}(y)^{2}v(0)dy = v(0) < \infty \,.
$$

(ii): With $A-1=(D^{1/2}-E^{1/2})E^{-1/2}$ and the integral
representation
\begin{align*}
  x^{1/4}-y^{1/4}= \sqrt{2}\pi
  \int_{0}^{\infty}t^{1/4}\left(\frac{1}{y+t} -\frac{1}{x+t}\right)dt
\end{align*}
we have
\begin{align} \nonumber
  \|A-1\|_{2} &\leq \sqrt{2}\pi    \int_{0}^{\infty}t^{{1/4}} \| \left((t+D^{2})^{-1}    -(t+E^{2})^{-1} \right)E^{-1/2}\|_{2} dt  \\
  & = 2^{3/2}\pi \int_{0}^{\infty} t^{1/4} \| (t+D^{2})^{-1}D^{1/2}VD^{1/2}(t+E^{2})^{-1}E^{-1/2}\|_{2}dt \nonumber \\
  &\leq 2^{3/2}\pi \|V\|_{2} \int_{0}^{\infty}t^{1/4}
  \|(t+D^{2})^{-1}D^{1/2}\| \|D^{1/2}(t+E^{2})^{-1}E^{-1/2}\| dt\,
  . \label{inta1}
\end{align}
Since $D\leq E$ we can further bound
\begin{align*}
  \|D^{1/2}(t+E^{2})^{-1}E^{-1/2}\|&=\|E^{-1/2}(t+E^{2})^{-1}D(t+E^{2})^{-1}E^{-1/2}\|^{1/2} \\
  &\leq \|(t+E^{2})^{-1}\| \leq \frac 1t \, .
\end{align*}
Using the spectral theorem we conclude that $
\|(t+D^{2})^{-1}D^{1/2}\| \leq C (1+t^{3/4})^{-1}$ and hence the
integrand in (\ref{inta1}) falls off like $~t^{-3/2}$ at infinity,
making the integral finite. The estimate for
$B-1=D^{-1/2}(E^{1/2}-D^{1/2})$ is obtained along the same lines.

(iii): We apply the integral representation $\log x = \frac 1 2
\int_{0}^{\infty}\left(\frac{1}{x+t}-\frac{1}{x^{-1}+t}\right)dt$ and
the resolvent identity to obtain
\begin{align*}
  2\|\alpha\|_{2}&\leq \int_{0}^{\infty}\|(|A^{{*}}|+t)^{-1}-(|B^{{*}}|+t)^{-1}\|_{2}dt \\
  & \leq \| |A^{{*}}|-|B^{{*}}| \|_{2}
  \int_{0}^{\infty}\|(|A^{{*}}|+t)^{-1}\| \|(|B^{{*}}|+t)^{-1}\|dt <
  \infty
\end{align*}
since $\| |A^{{*}}|-|B^{{*}}|\|_{2} = \|A-B\|_{2}\leq
\|A-1\|_{2}+\|B-1\|_{2}<\infty$ by (i).

(iv): On $ \mathcal{F}^{(N)}$ we have the bound
\begin{align*}
  \sum_{i,j}{}^{'}\overline{\alpha}_{ij}b^{\dagger}_{i}b_{j}^{\dagger}\sum_{i,j}{}^{'}{\alpha}_{ij}b_{i}b_{j} &\leq \frac{a_{0}^{2}(a_{0}^{\dagger})^{2}}{(N-1)^{2}} \sum_{i,j,k,l} \overline{\alpha}_{ij}\alpha_{kl} a_{i}^{\dagger}a_{j}^{\dagger}a_{k}a_{l} \\
  &\leq \left(\frac{N+2}{N-1}\right)^{2} \|\alpha\|^{2}_{2}N(N-1) \,,
\end{align*}
which shows that $X$ is a bounded operator. Its anti-symmetry follows
directly from its definition.
\end{proof}
This completes the proof of part (I) of
Proposition~\ref{maintechprop}.

\subsection{Proof of Proposition~\ref{maintechprop}
  (II)} \label{ss:II}

We abbreviate the symplectic transformation (\ref{invtrafo}) by
\begin{align}\label{notat}
  \left(\begin{array}{c} {c} \\ {c}^{\dagger} \end{array}\right) =:
  \left(\begin{array}{cc} L& M \\ M & L \end{array}\right)
  \left(\begin{array}{c} {b} \\ {b}^{\dagger} \end{array}\right)\, .
\end{align}
A straightforward computation shows that
\begin{align*}
  [c_{i},c_{j}^{\dagger}]+[c_{j},c_{i}^{\dagger}]&=2\frac{N-N^{>}}{N-1}\delta_{ij}-\frac{1}{N-1}\sum_{k,l}{}^{'}\left(L_{jl}L_{ik}-M_{jl}M_{ik}\right)
  \left(a_{k}^{\dagger}a_{l}+a_{l}^{\dagger}a_{k}\right)\, .
\end{align*}
We will show below that $Y- E/ 2$ and $D-E$ are trace class, with
\begin{align}\label{traceident}
  \mathrm{tr}\left(Y-\frac{E}{2}\right)=\frac 1 2 \Tr \left( D+ Q
    V-E\right) \, .
\end{align}
Given that, we have
\begin{align*}
  & 2 \sum_{i,j}{}^{'}  \Big(U_{0}^{{*}} \left( Y-\frac{E}{2}\right)U_{0}\Big)_{ij} [{c}_{i},{c}^{\dagger}_{j}]  \\
  & = \frac{N-N^{>}}{N-1}\mathrm{tr}\left( D+ Q V-E\right) \\
  & \quad -\frac{1}{N-1}\sum_{k,l}{}^{'}\left(L^{{*}}U_{0}^{{*}}\left(
      Y-\frac E 2 \right)U_{0}L-M^{{*}}U_{0}^{{*}}\left( Y-\frac E 2
    \right)U_{0}M \right)_{kl}
  \left(a_{k}^{\dagger}a_{l}+a_{l}^{\dagger}a_{k}\right)\, .
\end{align*}
Since $\Tr\, QV = \Tr \, V - v_{0000}$, Proposition~\ref{maintechprop}
(II) then follows if we can show that
\begin{align}\label{prophets}
  -C D \leq L^{{*}}U_{0}^{{*}}\left(Y-\frac E
    2\right)U_{0}L-M^{{*}}U_{0}^{{*}}\left(Y-\frac E 2\right)U_{0}M
  \leq C D\, .
\end{align}
We compute
\begin{align*}
  & 4L^{{*}}U_{0}^{{*}}\Big(Y- \frac E 2\Big)U_{0}L-4M^{{*}}U_{0}^{{*}}\Big(Y-\frac E 2 \Big)U_{0}M \\
  &= \frac 1 2 D^{{1/2}}\left(1+D^{-{1/2}}{V}D^{-{1/2}}+D^{-1}E (1+D^{-{1/2}}{V}D^{-{1/2}}) DE^{-1}- 2 D^{-1}E +\text{h.c.}\right)D^{{1/2}} \\
  &=:D^{{1/2}}RD^{{1/2}}\, .
\end{align*}
Now $R$ is a bounded operator since $DE^{-{1}}$ and $D^{-1}E$ are
bounded, which follows from
\begin{align}
  \|DE^{-1}\|^{2}&=\|DE^{-2}D\| \leq1\, \label{boundedop1}, \\
  \|D^{-1}E\|^{2}&= \|D^{-1}E^{2}D^{-1}\|= \|1 +2D^{-{1/2}}VD^{-{1/2}}
  \|<\infty\, . \label{boundedop2}
\end{align}
This proves (\ref{prophets}).

We now turn to (\ref{traceident}). Note that
\begin{align*}
  2\left(2Y- E \right)&=B^{{*}}\left({D}+V -D^{1/2}{E}D^{-1/2}
  \right)A+A^{{*}}\left( D + V - D^{-1/2}{E}D^{1/2} \right)B \, .
\end{align*}
We claim that
\begin{align}
  \label{tchs1}
  \|D^{1/2}(E-D)D^{-1/2}\|_{2}&<\infty \, , \\
  \label{tchs2}
  \|D^{1/2}(E-D)D^{-1/2}+\text{h.c.}\|_{1}&< \infty \, ,\\
  \|D+QV-E\|_{1}&< \infty \label{tchs0} \, .
\end{align}
Since by Lemma~\ref{hscondi} $A-1$, $B-1$, are Hilbert-Schmidt and, in
addition, $V$ is trace class by Lemma~\ref{hscondi}, it follows from
(\ref{tchs1})--(\ref{tchs0}) that
\begin{align*}
  B^{{*}}\left({D}+V -D^{1/2}{E}D^{-1/2} \right)A+\text{h.c.}={D}+QV
  -D^{1/2}{E}D^{-1/2} +\text{h.c.}+\mathrm{Rest}
\end{align*}
with $\|\mathrm{Rest}\|_{1}< \infty$; hence $2Y-E$ is trace
class. Moreover,
\begin{align*}
  \mathrm{tr}(2Y-E) &=\frac 1 2\mathrm{tr}\left(D^{1/2}(D-E)D^{-1/2}+\mathrm{h.c.}\right) + \mathrm{tr} \, Q V \\
  &= \mathrm{tr}\left(D+ {Q V}-E\right) \, ,
\end{align*}
where the first equality holds by cyclicity of the trace and the
second is seen to be true by computing the trace in the eigenbasis of
$D$.

To show (\ref{tchs1})--(\ref{tchs0}) we compute
\begin{align}\nonumber
  D^{1/2}(E-D)D^{-1/2}&= \pi D^{1/2} \int_{0}^{\infty}\sqrt{t}\left(
    ({t+D^{2}})^{-1}-({t+E^{2}})^{-1}\right)D^{-1/2}dt \\ \nonumber
  &=2\pi
  \int_{0}^{\infty}\sqrt{t}{D}({t+D^{2}})^{-1}VD^{1/2}({t+E^{2}})^{-1}
  D^{-1/2} dt \\ \nonumber
  &= 2\pi \int_{0}^{\infty}\sqrt{t}{D}({t+D^{2}})^{-1}V ({t+D^{2}})^{-1}dt  \\
  & \ \ \ \ \ - 4\pi
  \int_{0}^{\infty}\sqrt{t}{D}({t+D^{2}})^{-1}VD^{1/2}
  ({t+E^{2}})^{-1}D^{1/2}V({t+D^{2}})^{-1}dt \,, \label{4l}
\end{align}
where we applied the resolvent identity twice. The expression on the
last line is trace class. This follows from the bound
\begin{align*}
  \left\|{D}({t+D^{2}})^{-1}VD^{1/2}
    ({t+E^{2}})^{-1}D^{1/2}V({t+D^{2}})^{-1} \right\|_{1}\leq
  \left\|{D}({t+D^{2}})^{-1}\right\|^{2}\left\| ({t+D^{2}})^{-1}
  \right\| \left\| V\right\|_{1}^{2} \,,
\end{align*}
where we have used that $E^2\geq D^2$ in the second factor. The latter
expression falls off like $ t^{-2}$ for large $t$, making the integral
finite. For the first term on the right side of (\ref{4l}), we compute
its matrix elements. With $D_i = \epsilon_i -\epsilon_0$ the
eigenvalues of $D$,
\begin{align*}
  & \left\langle \varphi_i \left| \pi \int_{0}^{\infty}\sqrt{t}
      {D}({t+D^{2}})^{-1}V ({t+D^{2}})^{-1}dt \right| \varphi_{j}
  \right\rangle  \\
  &=V_{ij}\pi \int_{0}^{\infty}\sqrt{t}\frac{D_{i}}{t+D_{i}^{2}}
  \frac{1}{t+D_{j}^{2}}dt = V_{ij}\frac{D_{i}}{D_{i}+D_{j}}\, .
\end{align*}
In particular, since
\begin{align*}
  \left|V_{ij}\frac{D_{i}}{D_{i}+D_{j}}\right|\leq |V_{ij}| \,,
\end{align*}
the Hilbert-Schmidt property (\ref{tchs1}) follows. Moreover,
\begin{align*}
  V_{ij}\frac{D_{i}}{D_{i}+D_{j}}+(i \leftrightarrow j) = V_{ij} \,,
\end{align*}
which implies (\ref{tchs2}). To prove (\ref{tchs0}), one simply
computes the trace of the operator in (\ref{tchs2}) in the basis of
$D$, which leads to the conclusion that $\sum_i' \langle \varphi_i | E
- D| \varphi_i \rangle < \infty$.  Since $E-D$ is a positive operator,
this implies that $E-D$ is trace class. Since also $V$ is trace class,
this proves (\ref{tchs0}).

\subsection{Proof of Proposition~\ref{maintechprop} (III)}

Recall the notation introduced in (\ref{notat}). A straightforward
computation shows
\begin{align*} [c_{j},c_{i}]= \frac{1}{N-1}
  \sum_{k,l}{}^{'}\left(M_{jk}L_{il}a_{l}a_{k}^{\dagger}-L_{jk}M_{il}a_{k}a_{l}^{\dagger}\right)
\end{align*}
and
\begin{align*}
  \sum_{i,j}{}^{'}Z_{ij}\left([{c}_{j},{c}_{i}]+[{c}_{i}^{\dagger},{c}_{j}^{\dagger}]
  \right)=\frac{1}{N-1}\sum_{k,l}{}^{'}
  \left(L^{{*}}ZM-M^{{*}}ZL\right)_{kl}\left(a_{k}^{\dagger}a_{l}+a_{l}^{\dagger}a_{k}\right)\,
  .
\end{align*}
Hence what we need to show is
\begin{align*}
  -CD \leq L^{{*}}ZM-M^{{*}}ZL \leq CD\, .
\end{align*}
We observe that
\begin{align*}
  & 8\left(L^{{*}}ZM-M^{{*}}ZL\right) \\ &=\frac 1 2(B-A)\left(B^{{*}}\left(D+{V}\right)A-A^{{*}}\left(D+{V}\right)B\right)(A^{{*}}+B^{{*}}) +\text{h.c.} \\
  &=\left[D^{-{1/2}}ED^{-{1/2}}(D+V)D^{{1/2}}E^{-1}D^{{1/2}}-D-V\right]+\text{h.c.} \\
  &=D^{{1/2}}\Big(\Big[D^{-1}E(1+D^{-{1/2}}VD^{-{1/2}})DE^{-1}
  -1-D^{-{1/2}}VD^{-{1/2}} \Big]+\text{h.c.}\Big)D^{{1/2}}\, .
\end{align*}
The operator in square brackets is bounded because of
(\ref{boundedop1}) and (\ref{boundedop2}), hence the claim follows.

\section{Proof of Theorem~\ref{mainthm}}\label{finish}

This section contains the proof of Theorem~\ref{mainthm}. We split the
proof into two parts, corresponding to the lower and upper bounds on
the eigenvalues of $H_N$, respectively.

\subsection{Lower bound}

By combining the estimate~(\ref{HNlower}) with
Proposition~\ref{maintechprop}, we obtain the inequality
\begin{align} \nonumber H_{N}&\geq Nh_{00}+\frac{N+1}{2}v_{0000}+
  (1-\lambda)\mathcal{U}^{\dagger}\left(\sum_{i}{}^{'}e_{i}
    a_{i}^{\dagger}a_{i}\right)\mathcal{U}+\frac 1 2
  \mathrm{tr}(E-D-V)\\ \nonumber
  & \quad -C\Big( \left(N^{-1}\eps^{-1}+N^{-1}\lambda^{-1}+\zeta^{-1}N^{-1/2}\right)(N^{>}+1)(T_\text{H}+1)\\
  &\hspace{1in}+(N^{-1}+\eps+\zeta N^{-1/2})(T_\text{H}+1)\Big)\,
  , \label{lowerlast}
\end{align}
which holds for any $\lambda>0$, $\zeta>0$ and $0<\epsilon<1$.  Since
the spectrum of $\sum_{i}'e_{i}a_{i}^{\dagger}a_{i}$ consists of
finite sums of the form $\sum_{i}' e_{i}n_{i}$ with
$\sum_{i}'n_{i}\leq N$, the desired lower bound follows directly from
the min-max principle. In fact, for any function $\Psi$ in the
spectral subspace of $H_{N}$ corresponding to energy $E\leq
E_{0}(N)+\xi$, Lemmas~\ref{simple bounds} and~\ref{quadrbdlem} imply
that
$$
\bra{\Psi} (T_\text{H}+1) \ket{\Psi} \leq C(\xi+1)
$$
and
\begin{align*}
  \bra{\Psi} (N^{>}+1)(T_\text{H}+1) \ket{\Psi} \leq C(\xi+1)^{2} \, .
\end{align*}
Choosing $\eps = O(\sqrt{\xi/N})=\lambda$ and $\zeta=\sqrt{\xi}$, we
conclude that the spectrum of $H_{N}$ below an energy $E_{0}(N)+\xi$
is bounded from below by the corresponding spectrum of
\begin{align*}
  Nh_{00}+\frac{N+1}{2}v_{0000}+\sum_{i}{}^{'}e_{i}a_{i}^{\dagger}a_{i}-\frac
  1 2 \mathrm{tr}(D+V-E)-O(\xi^{\frac{3}{2}}N^{-1/2})\, .
\end{align*}
This completes the desired lower bound.

\subsection{Upper bound}\label{ss:ub}

A combination of (\ref{HNupper}) and Proposition~\ref{maintechprop}
implies that
\begin{align}\nonumber
  H_{N}& \leq Nh_{00}+\frac {N+1} 2
  v_{0000}+(1+\lambda)\mathcal{U}^{\dagger}\left(\sum_{i}{}'e_{i}a_{i}^{\dagger}a_{i}\right)\mathcal{U}-\frac
  1 2 \mathrm{tr}(D+V-E) \\ \nonumber
  & \quad +C \left(N^{-1}\eps^{-1}+N^{-1}\lambda^{-1}  +N^{-1}       +\zeta^{-1}N^{-1/2}\right)(N^{>}+1)(T_\text{H}+1)\, \\
  & \quad + C\left(\eps+\zeta N^{-1/2}+N^{-1}
  \right)(N^{>}+1)^{1/2}(T_\text{H}+1)^{1/2} \, , \label{hb1}
\end{align}
for any $\lambda>0$, $\zeta>0$ and $\epsilon>0$.  To apply the min-max
principle we need the following bound.
\begin{lem} \label{quadraticbound} One has the bound
  \begin{align} \label{lem6}
    \mathcal{U}(N^{>}+1)(T_\text{H}+1)\mathcal{U}^{\dagger} & \leq
    C\left(\sum_{i}{}^{'}e_{i}a_{i}^{\dagger}a_{i}+1\right)^{2} \, .
  \end{align}
\end{lem}

Note that by operator monotonicity of the square root it follows
immediately from Lemma~\ref{quadraticbound} that
\begin{align*}
  \mathcal{U}(N^{>}+1)^{1/2}(T_\text{H}+1)^{1/2}\mathcal{U}^{\dagger}
  & \leq C\left(\sum_{i}{}^{'}e_{i}a_{i}^{\dagger}a_{i}+1\right)\, .
\end{align*}
Hence we obtain from (\ref{hb1})
\begin{align}\label{upperestlast} \nonumber
  \mathcal{U}H_{N}\mathcal{U}^{\dagger} &\leq Nh_{00}+\frac{N+1}{2}v_{0000}+(1+\lambda)\sum_{i}{}^{'}e_{i}a_{i}^{\dagger}a_{i}-\frac 1 2 \mathrm{tr}(D+V-E) \\
  & \quad +C \left(N^{-1}\eps^{-1}+N^{-1}\lambda^{-1}  +N^{-1}       +\zeta^{-1}N^{-1/2}\right)\left(\sum_{i}{}^{'}e_{i}a_{i}^{\dagger}a_{i}+1\right)^{2}\, \nonumber \\
  & \quad + C \left(\eps+\zeta N^{-1/2}+N^{-1}
  \right)\left(\sum_{i}{}^{'}e_{i}a_{i}^{\dagger}a_{i}+1\right) \, .
\end{align}
Given an eigenvalue of $\sum_{i}^{'}e_{i}a_{i}^{\dagger}a_{i} $ with
value $\xi$, we choose $\eps =O(\sqrt{\xi/N})=\lambda$ and
$\zeta=\sqrt{\xi}$ to obtain $Nh_{00}+\frac {N+1} 2 v_{0000}+\xi
-\frac 1 2\mathrm{tr}(D+V-E)+O(\xi^{3/2}N^{-1/2})$ for the right side
of (\ref{upperestlast}). Hence the desired upper bound follows from
the min-max principle.

It remains to prove (\ref{lem6}).

\begin{proof}[Proof of Lemma~\ref{quadraticbound}] If we can show that
  \begin{align} \label{firststep} \e^{X}(N^{>}+1)(T_\text{H}+1)\e^{-X}
    \leq C (N^{>}+1)(T_\text{H}+1)\,
  \end{align}
  and
  \begin{align}\label{secondstep}
    W^{*}DW=U_0^{*}W_{0}^{*} DW_{0}U_0 \leq C \hat{E} \,,
  \end{align}
  the claim follows since then
  \begin{align*}
    \mathcal{U}(N^{>}+1)(T_\text{H}+1)\mathcal{U}^{\dagger}&\leq C\mathcal{W}^{\dagger}(N^{>}+1)^{1/2}(T_\text{H}+1)(N^{>}+1)^{1/2}\mathcal{W} \\
    &=C(N^{>}+1)^{1/2}\mathcal{W}^{\dagger}(T_\text{H}+1)\mathcal{W}(N^{>}+1)^{1/2} \\
    &\leq C
    (N^{>}+1)\left(\sum_{i}{}^{'}e_{i}a_{i}^{\dagger}a_{i}+1\right) \,
    ,
  \end{align*}
  where we have used (\ref{firststep}) for the first inequality, and
  (\ref{secondstep}) for the second.

  We start with the proof of (\ref{firststep}). In fact we shall show
  that
  \begin{align}\label{operatorest}
    \e^{X}(N^{>}+1)^{2}(T_\text{H}+1)^{2}\e^{-X} \leq C
    (N^{>}+1)^{2}(T_\text{H}+1)^{2}
  \end{align}
  from which the claim follows by operator monotonicity of the square
  root. We compute
  \begin{align} \nonumber
    [X,(N^{>}+1)^{2}(T_\text{H}+1)^{2}] =& (N^{>}+1)(T_\text{H}+1)[X,(N^{>}+1)(T_\text{H}+1)] \\
    &+[X,(N^{>}+1)(T_\text{H}+1)](N^{>}+1)(T_\text{H}+1)\,
    . \label{commutX}
  \end{align}
  With
  \begin{align*}
    A_{1}:=[X,N^{>}]&= \sum_{i,j}{}^{'}\alpha_{ij}\left(b_{i}b_{j}+b_{j}^{\dagger}b_{i}^{\dagger}\right) \\
    A_{2}:=[X,T_\text{H}]&=\sum_{i,j}{}^{'}\alpha_{ij}(\eps_{i}-\eps_{0})\left(b_{i}b_{j}+b_{j}^{\dagger}b_{i}^{\dagger}\right)
  \end{align*}
  we can bound
  \begin{align*}
    [X,(N^{>}+1)(T_\text{H}+1)]^{2}&= \left(A_{1}(T_\text{H}+1)+(N^{>}+1)A_{2}\right)^{2} \\
    &= \left(A_{1}(T_\text{H}+1)+A_{2}(N^{>}+1)+[N^{>},A_{2}]\right)^{2} \\
    &\leq C
    \left((T_\text{H}+1)A_{1}^{2}(T_\text{H}+1)+(N^{>}+1)A_{2}^{2}(N^{>}+1)
      + [N^{>},A_{2}]^{2}\right)\, .
  \end{align*}
  By (\ref{squarecommest}) we have
  \begin{align*}
    A_{1}^{2}\leq C \|\alpha\|_{2}^{2} (N^{>}+1)^{2}
  \end{align*}
  and similarly
  \begin{align*}
    A_{2}^{2} \leq C \|D\alpha\|_{2}^{2}(N^{>}+1)^{2}\, .
  \end{align*}
  Furthermore, since
  \begin{align*}
    [N^{>},A_{2}]=2\sum_{i,j}{}^{'}\alpha_{ij}(\eps_{i}-\eps_{0})\left(b_{j}^{\dagger}b_{i}^{\dagger}-b_{i}b_{j}\right)
  \end{align*}
  one checks that
  \begin{align*} [N^{>},A_{2}]^{2}\leq C
    \|D\alpha\|_{2}^{2}(N^{>}+1)^{2}\, .
  \end{align*}

  To see that $\|D\alpha \|_{2}<\infty$, we can proceed as in
  (\ref{logestA}) and bound
  \begin{align*}
    D\alpha^{2}D\leq
    D(D^{-1/2}ED^{-1/2}-1)^{2}D=D^{1/2}(E-D)D^{-1}(E-D)D^{1/2}\,.
  \end{align*}
  Hence we have
  \begin{align*}
    \|D\alpha\|_{2} & \leq \|D^{1/2}(E-D)D^{-1/2}\|_{2}
  \end{align*}
  which is finite due to (\ref{tchs1}). Applying Schwarz's inequality
  to (\ref{commutX}), we have thus shown that
  \begin{align*} [X,(N^{>}+1)^{2}(T_\text{H}+1)^{2}] & \leq C
    (N^{>}+1)^{2}(T_\text{H}+1)^{2}\,.
  \end{align*}
  We further have
  \begin{align*}
    \e^{tX}(N^{>}+1)^{2}(T_\text{H}+1)^{2}\e^{-tX} & =
    (N^{>}+1)^{2}(T_\text{H}+1)^{2} \\ & \quad + \int_0^t e^{sX}
    [X,(N^{>}+1)^{2}(T_\text{H}+1)^{2}] e^{-sX} dx \\ & \leq
    (N^{>}+1)^{2}(T_\text{H}+1)^{2} \\ & \quad + C
    \int_{0}^{t}\e^{sX}(N^{>}+1)^{2}(T_\text{H}+1)^{2}\e^{-sX}ds
  \end{align*}
  which by Gr\"onwall's inequality implies (\ref{operatorest}).

  For the proof of (\ref{secondstep}) we need to show that
  \begin{align*}
    W_{0}^{*}DW_{0}\leq CE
  \end{align*}
  or, equivalently, that
  \begin{align}
    D^{1/2}W_{0}E^{-1/2}&=DE^{-1/2}(E^{-1/2}DE^{-1/2})^{-1/2}E^{-1/2} \nonumber \\
    &=\pi D\int_{0}^{\infty}t^{-1/2}(Et+D)^{-1}dt \label{reint}
  \end{align}
  is a bounded operator. Observe that, by (\ref{boundedop1}),
  \begin{align}\label{nonsares}
    \left\|D (Et+D)^{-1} \right\|\leq \|DE^{-1}\| \|Q
    (t+DE^{-1})^{-1}\|\leq \| Q (t+DE^{-1})^{-1}\|\, .
  \end{align}
  With the aid of a Neumann expansion, one sees that the right side of
  (\ref{nonsares}) can be bounded by $2t^{-1}$ for $t> 2\|DE^{-1}\| $,
  which gives a bounded contribution to the integral in
  (\ref{reint}). For $t\leq 2 \|DE^{-1}\|$, one can argue that by
  analyticity of the resolvent map $t \mapsto (t+DE^{-1})^{-1}$, as
  well as the fact that $ED^{-1}$ is bounded, we get a uniform bound
  on $\|Q (t+DE^{-1})^{-1}\|$. This argument does not yield a
  quantitative bound, however, since $DE^{-1}$ is not a self-adjoint
  operator. To obtain an explicit bound, we make use of the fact that
  $DE^{-1}-1$ is a Hilbert-Schmidt operator. In fact, it is even trace
  class, since by (\ref{boundedop1}) and (\ref{tchs0})
  \begin{align*}
    \|D E^{-1} - 1 \|_{1} &=\|DE^{-1}(D-E)D^{-1}\|_{1} \leq
    \|(D-E)D^{-1}\|_{1} < \infty \, .
  \end{align*}
  We shall apply the following result.

  \begin{lem}[Theorem~6.4.1 in \cite{Gil}]\label{Gillemma} Let $A$ be
    a Hilbert-Schmidt operator. Then for $z\notin \sigma(A)$ (the
    spectrum of $A$)
    \begin{align*}
      \|(A-z)^{-1}\| \leq
      \sum_{k=0}^{\infty}\frac{\|A\|_{2}^{k}}{(\inf_{t\in
          \sigma(A)}|z-t|)^{k+1}\sqrt{k!}}\, .
    \end{align*}
  \end{lem}

  Define $a$ to be the infimum of the spectrum of $DE^{-1}$ on the
  space $Q \mathcal{F}^{(1)}$. It equals the infimum of the spectrum
  of $E^{-1/2} D E^{-1/2}$ on that space, hence
$$
a = \| E^{1/2} D^{-1} E^{1/2} \| ^{-1} > 0 \,.
$$
By Lemma~\ref{Gillemma} we thus have
\begin{align*}
  \|Q(t+DE^{-1})^{-1}\|&=\|Q(t+1+DE^{-1}-1)^{-1} \|\\
  &\leq \sum_{k=0}^{\infty} \frac{\|DE^{-1}-1\|_{2}^{k}}{ (t+a)^{k+1} \sqrt{k!}   }\\
  & \leq \frac{\sqrt{2}}{t+a} \exp \left(\frac{\|DE^{-1}-1\|_{2}}{t+a}
  \right)\, .
\end{align*}
Here we have used the bound $\sum_{k=0}^{\infty} x^{k}/\sqrt{k!}\leq
\sqrt{2} \e^{x^{2}}$ for $x\geq 0$ (cf. p. 84 in \cite{Gil}). This
yields the desired quantitative bound, and concludes the proof of the
boundedness of (\ref{reint}).
\end{proof}

\section{Consequences for Eigenvectors} \label{corproof}

\subsection{Proof of Corollary~\ref{cor2}}

We abbreviate
\begin{align*}
  H&:= H_{N}- E_0(N) + 1
  =:\sum_{i=1}^{\infty}h_{i}\ket{\chi_{i}}\bra{\chi_{i}}\, ,
\end{align*}
with $h_{i}\leq h_{i+1}$. For $h_j\leq \xi$, it follows from
(\ref{lowerlast}) and Lemmas~\ref{simple bounds}--\ref{quadrbdlem}
that
$$
\bra{\chi_{j}} K \ket{\chi_{j}} \leq h_{j}\left(1+C(\xi/N)^{1/2}
\right)\, .
$$
From (\ref{upperestlast}) we further deduce that $h_j \leq k_j
\left(1+C(\xi/N)^{1/2} \right)$, and thus
\begin{align*}
  \bra{\chi_{j}} K \ket{\chi_{j}} \leq k_{j}\left(1+C(\xi/N)^{1/2}
  \right)\, .
\end{align*}
A simple application of the min-max principle \cite[Lemma~2]{yinsei}
then shows that if $k_{j+1}>k_{j}$ then
\begin{align*}
  \sum_{k,l=1}^{j}|\langle \chi_{k},\psi_{l}\rangle|^{2}\geq j-
  C(\xi/N)^{1/2}\frac{ \sum_{l=1}^{j}k_{l}}{k_{j+1}-k_{j}}\, .
\end{align*}
In other words, with
$P^{j}_{K}:=\sum_{k=1}^{j}\ket{\psi_{k}}\bra{\psi_{k}}$ and
$P^{j}_{H}:=\sum_{k=1}^{j}\ket{\psi_{k}}\bra{\psi_{k}}$,
\begin{align*}
  \| P_{K}^{j}-P_{H}^{j}\|_{2}^{2}\leq C(\xi/N)^{1/2}\frac{
    \sum_{l=1}^{j}k_{l}}{k_{j+1}-k_{j}}\, .
\end{align*}
This completes the proof. \hfill\qed
 
\begin{rem}\label{rem:pr}
  Note that the (normalized) eigenfunctions of $K$ can be written as
  \begin{align} \label{close} \left(\mathcal{U}^{\dagger} \prod_{i\geq
        1}\frac{(a_{i}^{\dagger})^{n_{i}}}{\sqrt{n_{i}!}}\mathcal{U}\right)
    \mathcal{U}^{\dagger}\ket{ N-n,0,\dots}=\prod_{i\geq
      1}\frac{(d_{i}^{\dagger}+K_{i}^{\dagger})^{n_{i}}}{\sqrt{n_{i}!}}
    \mathcal{U}^{\dagger}\ket{ N-n,0,\dots}
  \end{align}
  where $n=\sum_{i\geq 1} n_{i}\leq N$, and $\ket{N-n,0,\dots}$
  denotes the function $\otimes_{i=1}^{N-n} \varphi_0 \in
  \mathcal{F}^{(N-n)}$. The operators $d_i$ are explicitly defined in
  (\ref{def:di}). The operators $K_i$ are small in the low-energy
  subspace, as shown in the proof of Lemma~\ref{secondsteplem}. The
  eigenfunctions of $K$ (and, hence, the ones of $H_N$) are thus
  approximately obtained by applying the raising-type operators
  $d_i^\dagger$ to the $N-n$-particle ground state.  To explicitly
  estimate the difference of the functions (\ref{close}) and
  \begin{align*}
    \prod_{i\geq 1}\frac{(d_{i}^{\dagger})^{n_{i}}}{\sqrt{n_{i}!}}
    \mathcal{U}^{\dagger}\ket{ N-n,0,\dots}\,,
  \end{align*}
  however, it would be necessary to give bounds on products of powers
  of the operators $K_{i}^{\dagger}$ and $d_{i}^{\dagger}$, which are
  more involved than the ones used in Lemma~\ref{secondsteplem}.
\end{rem}

\begin{rem}\label{finrem}
  As noted in Section~\ref{ss:mod}, Corollary~\ref{cor2} implies that
  the ground state $\Psi_{0}$ of $H_N$ is close, in $L^2$-norm, to
  $\mathcal{U}^\dagger\ket{N,0,\dots}$. To see the importance of the
  unitary operator $\mathcal{U}$, one can calculate the matrix element
  \begin{align}
    \bra{N,0,\dots}\mathcal{U}^{\dagger}\ket{
      N,0,\dots}=\bra{N,0,\dots}\e^{-X}\ket{ N,0,\dots} \,
    . \label{overlap}
  \end{align}
  This equality follows from the fact that $W$ leaves the Hartree
  ground state $\varphi_{0}$ invariant.  One readily checks that
  $\frac{d}{dt} \bra{N,0,\dots}\e^{-tX}\ket{ N,0,\dots} |_{t=0}=0
  $. However,
$$
\frac{d^{2}}{dt^{2}} \bra{N,0,\dots}\e^{-tX}\ket{ N,0,\dots}
|_{t=0}=\bra{N,0,\dots}X^{2}\ket{ N,0,\dots} = -\frac N{2(N-1)}
\|\alpha\|_2^2 \,,
$$
which is not small for large $N$. Hence we expect that the matrix
element (\ref{overlap}) differs significantly from $1$.
\end{rem}

\vspace{.2in}

\noindent \emph{Acknowledgments.} It is a pleasure to thank Sven
Bachmann for helpful discussions. Financial support by NSERC is
gratefully acknowledged.

\end{document}